\documentclass[reprint,superscriptaddress,nofootinbib,nobibnotes,amsmath,amssymb,pra]{revtex4-2}

\usepackage{amsthm}

\usepackage{graphicx}% Include  figure files
\usepackage{dcolumn}% Align table columns on decimal point
\usepackage{bm}% bold math
\usepackage[colorlinks=true,linkcolor=blue,urlcolor=blue,citecolor=blue,breaklinks=true]{hyperref}% add hypertext capabilities

\usepackage{mathtools}
\usepackage{amsthm}
\usepackage{xcolor}
\usepackage{comment}
\usepackage{fancybox}
\usepackage{stmaryrd}
\usepackage{ytableau}
\usepackage{dsfont}
\usepackage{physics}
\usepackage{multirow}
\newtheorem{theorem}{Theorem}

\newtheorem{lemma}[theorem]{Lemma}

\newtheorem{corollary}[theorem]{Corollary}

\newcommand{\mcA}{\mathcal{A}}
\newcommand{\mcB}{\mathcal{B}}
\newcommand{\mcC}{\mathcal{C}}

\newcommand{\mcE}{\mathcal{E}}

\newcommand{\mcH}{\mathcal{H}}
\newcommand{\mcI}{\mathcal{I}}

\newcommand{\mcL}{\mathcal{L}}

\newcommand{\mcO}{\mathcal{O}}
\newcommand{\mcP}{\mathcal{P}}

\newcommand{\mcR}{\mathcal{R}}
\newcommand{\mcS}{\mathcal{S}}

\newcommand{\mcU}{\mathcal{U}}
\newcommand{\mcV}{\mathcal{V}}

\newcommand{\mcX}{\mathcal{X}}
\newcommand{\mcY}{\mathcal{Y}}
\newcommand{\mcZ}{\mathcal{Z}}

\newcommand{\mfS}{\mathfrak{S}}

\newcommand{\U}{\mathrm{U}}

\newcommand{\CC}{\mathbb{C}}

\newcommand{\dket}[1]{\vert {#1} \rangle\!\rangle}

\newcommand{\dketbra}[1]{\vert {#1} \rangle\!\rangle\!\langle\!\langle {#1} \vert}

\newcommand{\1}{\mathds{1}}

\begin{document}

\preprint{APS/123-QED}
\title{Random dilation superchannel}

\author{Satoshi Yoshida}
\email{satoshiyoshida.phys@gmail.com}
\affiliation{Department of Physics, Graduate School of Science, The University of Tokyo, Hongo 7-3-1, Bunkyo-ku, Tokyo 113-0033, Japan}
\author{Ryotaro Niwa}
\affiliation{Department of Physics, Graduate School of Science, The University of Tokyo, Hongo 7-3-1, Bunkyo-ku, Tokyo 113-0033, Japan}
\author{Takeru Utsumi}
\affiliation{Graduate School of Arts and Sciences, The University of Tokyo, 3-8-1 Komaba, Meguro-ku, Tokyo 153-8902, Japan}
\author{Ryuji Takagi}
\affiliation{Graduate School of Arts and Sciences, The University of Tokyo, 3-8-1 Komaba, Meguro-ku, Tokyo 153-8902, Japan}
\author{Mio Murao}
\affiliation{Department of Physics, Graduate School of Science, The University of Tokyo, Hongo 7-3-1, Bunkyo-ku, Tokyo 113-0033, Japan}
\affiliation{Trans-scale Quantum Science Institute, The University of Tokyo, Hongo 7-3-1, Bunkyo-ku, Tokyo 113-0033, Japan}
\date{\today}

\begin{abstract}
  We present a quantum circuit that implements the \emph{random dilation superchannel}, transforming parallel queries of an unknown quantum channel into the same number of parallel queries of a randomly chosen dilation isometry of the input channel.
  This is a natural generalization of the random purification channel, that transforms copies of an unknown mixed state to copies of a randomly chosen purification state.
  The circuit complexity of our construction is $O(\mathrm{poly}(n, \log d_{\mcI}, \log d_{\mcO}))$, where $n$ is the number of queries and $d_{\mcI}$ and $d_{\mcO}$ are the input and output dimensions of the input channel, respectively.
  This random dilation superchannel is extended to the sequential queries approximately, by transforming the parallel random dilation isometry into sequential random dilation unitaries with $O(\mathrm{poly}(d_{\mcI}))$ overhead in the number of queries.
  We also show that our results can be further extended to the case of quantum superchannels.
  On the other hand, we show a no-go theorem on the exact random dilation of sequential queries with $o(\mathrm{poly}(\min\{d_{\mcI}, d_{\mcO}\}))$ query overhead, showcasing a fundamental difference between the parallel and sequential cases.
  As an application, we show an efficient storage-and-retrieval of an unknown quantum channel, which improves the program cost exponentially in the retrieval error $\varepsilon$.
  For the case where the Kraus rank $r$ is the least possible (i.e., $r = d_{\mcI}/d_{\mcO}$), we show quantum circuits that transform $n$ parallel queries of an unknown quantum channel $\Lambda$  to $\Theta(n^\alpha)$ parallel queries of $\Lambda$ for any $\alpha<2$ approximately, and implement its Petz recovery map for the maximally mixed reference state probabilistically and exactly.
\end{abstract}

\maketitle

\section{Introduction}

Any quantum state $\rho\in \mcL(\CC^d)$\footnote{In this work, $\mcL(\mcV)$ denotes the set of linear operators acting on the vector space $\mcV$.} can be represented as a reduced state of a pure state $\ket{\psi} \in \CC^d \otimes \CC^d$ called a purification of $\rho$, i.e., $\rho = \Tr_{\CC^d}[\ketbra{\psi}]$.
Although universal transformation from a mixed state to its purification is impossible due to the no-purification theorem~\cite{liu2025no}, it is possible to construct a quantum channel that transforms $n$ copies of an unknown quantum state $\rho$ into $n$ copies of a \emph{randomly chosen} purification of $\rho$~\cite{soleimanifar2022testing,chen2024local,tang2025conjugate}.
We call such a quantum channel the \emph{random purification channel} following Ref.~\cite{girardi2025random}.
This technique is used to lift the distinguishability of mixed states to that of bipartite pure states~\cite{soleimanifar2022testing}, lift the property testing of mixed states to that of bipartite pure states~\cite{chen2024local}, lift the sample access to the query access of a state preparation unitary~\cite{tang2025conjugate}, realize smaller sampling cost for fidelity estimation~\cite{utsumi2025quantum}, construct the optimal mixed state tomography~\cite{pelecanos2025mixed} and channel tomography~\cite{mele2025optimal}, and show Uhlmann's theorem for quantum divergences~\cite{girardi2025random}, with extensions to the fermionic and bosonic systems~\cite{walter2025random,mele2025random}, showcasing the power of random purification in analyzing various operational tasks.

This series of developments motivates us to extend the notion of random purification to quantum channels.
Any quantum channel $\Lambda: \mcL(\mcI) \to \mcL(\mcO)$ admits its Stinespring representation $\Lambda(\cdot) = \Tr[V(\cdot)V^\dagger]$ using an isometry $V: \mcI \to \mcE \otimes \mcO$ called the dilation isometry and an auxiliary system $\mcE$.
Since the Stinespring dilation can be considered as a natural generalization of the purification, we consider the following task: given $n$ queries of an unknown quantum channel $\Lambda:\mcL(\mcI) \to \mcL(\mcO)$, implement $n$ queries of a randomly chosen dilation isometry of $\Lambda$.
We call such a universal transformation \emph{random dilation superchannels}, where superchannels are universal transformations of a channel to another channel implementable within quantum mechanics~\cite{zyczkowski2008quartic,chiribella2008transforming,chiribella2008quantum,oreshkov2012quantum}.
Since processing unknown quantum channels in a quantum circuit is often more challenging than processing their Stinespring dilation~\cite{gilyen2022quantum,Niwa2025udt}, it is of great interest to decide if such a random dilation superchannel exists and is efficiently implementable.
Indeed, the feasibility of constructing random dilation superchannels was posed as an important open problem in Ref.~\cite{tang2025conjugate}.
In Ref.~\cite{chen2025quantum}, a transformation of a quantum tester, a special type of quantum superchannels that transform a quantum channel to a measurement outcome, is presented.
However, it is difficult to assess whether it is efficiently implementable in a quantum circuit, since it relies on an existence proof in representation theory, and no explicit quantum circuit is provided.

This work provides a partial answer to the above open problem by constructing an efficient quantum circuit implementing the random dilation superchannel that transforms $n$ parallel queries of an unknown quantum channel $\Lambda$ into $n$ parallel queries of a randomly chosen dilation isometry of $\Lambda$.
Our construction is based on the quantum Fourier transform over the symmetric group and the quantum Schur transform, both of which can be efficiently implemented~\cite{beals1997quantum,kawano2016quantum,bacon2005quantum,bacon2006efficient,burchardt2025high}.
By using this random dilation superchannel for the unitary superreplication algorithm~\cite{chiribella2015universal}, we show a superreplication algorithm of an unknown quantum channel $\Lambda: \mcL(\mcI) \to \mcL(\mcO)$ with Kraus rank at most $r=d_{\mcI}/d_{\mcO}$ that transforms $n$ parallel queries of $\Lambda$ to $\Theta(n^\alpha)$ parallel queries of $\Lambda$ for any $\alpha<2$.
We also construct a quantum circuit implementing the Petz recovery map~\cite{Petz1986sufficientsubalgebas} of an unknown quantum channel $\Lambda$ with at most Kraus rank $r = d_{\mcI}/d_{\mcO}$ probabilistically and exactly based on the probabilistic unitary inversion algorithm~\cite{quintino2019reversing,quintino2019probabilistic}.
We also extend our results to further higher-order cases from states and channels, namely, the random dilation of quantum superchannels. We construct a quantum circuit implementing the random dilation supersuperchannel that transforms $n$ parallel queries of an unknown quantum superchannel into $n$ parallel queries of a randomly chosen dilation superchannel of the input superchannel.
By using a similar idea to the sample-to-query lifting of Ref.~\cite{tang2025conjugate}, we also show an approximate implementation of the sequential random dilation superchannel by transforming the parallel random dilation isometry into sequential random dilation unitaries with $O(\mathrm{poly}(d_{\mcI}))$ overhead in the number of queries, which is exponential in the number of qubits in the input system.

On the other hand, we also show that the exact random dilation of sequential queries of an unknown quantum channel is \emph{impossible} with subexponential overhead $o(\mathrm{poly}(\min\{d_\mcI, d_\mcO\}))$ in the number of queries,
See Tab.~\ref{tab:summary} for the summary of the results.
This no-go theorem is based on the no-go theorem of simulating the quantum switch by a quantum circuit whose causal order is classically controlled~\cite{bavaresco2025simulating}.
By showing that the random dilation of sequential queries of an unknown quantum channel can be used to simulate the quantum switch~\cite{chiribella2013quantum}, we conclude that the random dilation of sequential queries is impossible.
This result partially disproves the conjecture in Ref.~\cite{tang2025conjugate} that the random dilation of any type of queries possible with polynomial queries, and it also shows a fundamental difference between the parallel and sequential cases of the random dilation.
It also contrasts with the random purification channel, where the parallel and sequential cases are equivalent.

\begin{table}
    \centering
    \begin{ruledtabular}
    \begin{tabular}{c|c|c}
         & Exact & Approximate \\\hline
        Parallel & No overhead~\textbf{[Thm.~\ref{thm:random-dilation-superchannel}]} & --\\\hline
        \multirow{2}{*}{Sequential} & $\Omega(\mathrm{poly}(\min\{d_\mcI, d_\mcO\}))$ & $O(\mathrm{poly}(d_\mcI))$\\
         & \textbf{[Thm.~\ref{thm:arbitrary_dim_no-go}]} & \textbf{[Thm.~\ref{thm:random-dilation-superchannel} + Thm.~\ref{thm:lifting_parallel_to_sequential}]}
    \end{tabular}
    \end{ruledtabular}
    \caption{The required overhead in the number of queries to implement parallel or sequential random dilation superchannels of $\Lambda: \mcL(\mcI) \to \mcL(\mcO)$.
    Theorem~\ref{thm:random-dilation-superchannel} shows that the parallel one does not require any overhead even in an exact case.
    Theorem~\ref{thm:lifting_parallel_to_sequential} shows an approximate conversion from parallel dilation isometries to sequential dilation unitaries, which shows a sequential random dilation superchannel with $O(\mathrm{poly}(d_\mcI))$ overhead combined with Thm~\ref{thm:random-dilation-superchannel}.
    Theorem~\ref{thm:arbitrary_dim_no-go} shows a no-go theorem that any exact sequential random dilation superchannel at least requires $\Omega(\mathrm{poly}(\min\{d_\mcI, d_\mcO\}))$ overhead.}
    \label{tab:summary}
\end{table}

The rest of this paper is organized as follows.
In Sec.~\ref{sec:preliminaries}, we summarize the representation theory and quantum superchannels.
In Sec.~\ref{sec:parallel_random_dilation_superchannel}, we present the quantum circuits implementing the random dilation superchannel.
By constructing a random purification channel based on the quantum Fourier transform over the symmetric group and the quantum Schur transform in Sec.~\ref{subsec:random_purification_channel}, we construct the random dilation superchannel and the random dilation supersuperchannel in Sec.~\ref{subsec:parallel_random_dilation_superchannel}.
We also construct another quantum circuit of the parallel random dilation superchannel based on the Clebsch--Gordan transform over the symmetric group in Sec.~\ref{sec:another_construction}.
We present an approximate transformation of parallel random dilation isometries to sequential random dilation unitaries in Sec.~\ref{subsec:lifting_parallel_to_sequential} to obtain an approximate implementation of the sequential random dilation superchannel.
In Sec.~\ref{sec:sequential_random_dilation}, we show the no-go theorem on the exact random dilation of sequential queries of an unknown quantum channel.
In Sec.~\ref{sec:applications}, we present three applications of the random dilation superchannel: storage and retrieval of quantum channels, superreplication of quantum channels,  and probabilistic implementation of the Petz recovery map of quantum channels.
Finally, we conclude this work in Sec.~\ref{sec:conclusion}.

\section{Preliminaries}
\label{sec:preliminaries}

We summarize the essential notions of representation theory of the unitary group and the symmetric group, and quantum superchannels.
Readers familiar with representation theory can skip Sec.~\ref{subsec:representation_theory}, and readers familiar with quantum superchannels can skip Secs.~\ref{subsec:quantum_superchannels} and \ref{subsec:dilation_of_channels_and_combs}.

\subsection{Representation theory of the unitary group and the symmetric group}
\label{subsec:representation_theory}
We first introduce the Schur--Weyl duality, which is the main technical tool used in this work.
We consider the representation of the unitary group $\U(d)$ and the symmetric group $\mfS_n$ on the $n$-fold tensor product space $(\CC^d)^{\otimes n}$ given by
\begin{align}
  U^{\otimes n} \ket{i_1} \otimes \cdots \otimes \ket{i_n} &= U\ket{i_1} \otimes \cdots \otimes U\ket{i_n},\\
  \pi(\sigma) \ket{i_1} \otimes \cdots \otimes \ket{i_n} &\coloneqq \ket{i_{\sigma^{-1}(1)}} \otimes \cdots \otimes \ket{i_{\sigma^{-1}(n)}}.
\end{align}
These representations can be simultaneously decomposed into irreducible representations (irreps) as
\begin{align}
  \label{eq:schur_weyl_duality}
  (\CC^d)^{\otimes n} &\simeq \bigoplus_{\lambda \vdash n} \mcU_{\lambda} \otimes \mcS_{\lambda},\\
  U^{\otimes n} &\simeq \bigoplus_{\lambda \vdash n} f_{\lambda}(U) \otimes \1_{\mcS_{\lambda}},\\
  \pi(\sigma) &\simeq \bigoplus_{\lambda \vdash n} \1_{\mcU_{\lambda}} \otimes g_{\lambda}(\sigma),
\end{align}
where $\lambda \vdash n$ means that $\lambda$ is a partition of $n$, $f_\lambda: \U(d) \to \mcL(\mcU_{\lambda})$ is the irrep of $\U(d)$ labeled by $\lambda$, $g_\lambda:\mfS_n \to \mcL(\mcS_{\lambda})$ is the irrep of $\mfS_n$ labeled by $\lambda$, and $\1_\mcH$ represents the identity operator on a Hilbert space $\mcH$.
The dimensions of these irrep spaces are denoted by $d_{\lambda} \coloneqq \dim \mcU_{\lambda}$ and $m_{\lambda} \coloneqq \dim \mcS_{\lambda}$.
We define an orthonormal basis of $\mcU_{\lambda}$ and $\mcS_{\lambda}$ as $\{\ket{u}\}_{u\in [d_\lambda]}$ and $\{\ket{i}\}_{i\in [m_\lambda]}$, respectively, and we define the Schur basis of $(\CC^d)^{\otimes n}$ as $\{\ket{\lambda, u, i} \coloneqq \ket{u}\otimes \ket{i}\}_{\lambda\vdash n, u\in [d_\lambda], i \in [m_\lambda]}$.
We in particular take the Gelfand--Tsetlin basis~\cite{goodman2009symmetry} as the orthonormal basis of $\mcU_{\lambda}$, and the Young--Yamanouchi basis~\cite{james2006representation} as the orthonormal basis of $\mcS_{\lambda}$.
The quantum Schur transform $U_\mathrm{Sch}^{(n,d)}$ is the unitary transformation representing the basis change from the computational basis to the Schur basis $\ket{\lambda, u, i}$.
The quantum Schur transform can be efficiently implemented with the circuit complexity $O(\mathrm{poly}(n, \log d))$~\cite{burchardt2025high}.

We then introduce the quantum states and channels used in the construction of the parallel dilation quantum superchannel.
We define the quantum state in the group algebra $\CC[\mfS_n] \coloneqq \mathrm{span}_{\CC}\{\ket{\sigma} \mid \sigma \in \mfS_n\}$ as
\begin{align}
  \label{eq:def_plus}
  \ket{+_{\mfS_n}} \coloneqq {1\over \sqrt{n!}} \sum_{\sigma \in \mfS_n} \ket{\sigma} \in \CC[\mfS_n],
\end{align}
the controlled permutation unitary as
\begin{align}
  \mathrm{ctrl}-\pi \coloneqq \sum_{\sigma \in \mfS_n} \ketbra{\sigma}{\sigma} \otimes \pi(\sigma),
\end{align}
the quantum Fourier transform (QFT) over the symmetric group $\mfS_n$ as
\begin{align}
  \label{eq:def_qft}
  \mathrm{QFT}_{\mfS_n} \colon \CC[\mfS_n] &\to \bigoplus_{\lambda \vdash n} \mcS_{\lambda} \otimes \mcS_{\lambda},\\
  \ket{\sigma} &\mapsto \sum_{\lambda \vdash n} \sum_{i, j=1}^{m_\lambda} \sqrt{m_\lambda \over n!} [g_{\lambda}(\sigma)]_{ji} \ket{\lambda, i, j},
\end{align}
and the maximally mixed state in $\mcU_\lambda$ by
\begin{align}
    \label{eq:def_pi_u}
    \pi_{\mcU_\lambda}\coloneqq {\1_{\mcU_\lambda}\over d_\lambda}.
\end{align}
The unitarity of $\mathrm{QFT}_{\mfS_n}$ is guaranteed by the Schur orthogonality relation
\begin{align}
\label{eq:schur_orthogonality}
  {1\over n!} \sum_{\sigma \in \mfS_n} [g_{\mu}(\sigma)]_{ji} [\overline{g_{\lambda}(\sigma)}]_{j'i'} = {1\over m_\lambda} \delta_{\lambda\mu}\delta_{j j'} \delta_{i i'},
\end{align}
where $\delta_{ii'}$ is the Kronecker delta defined by $\delta_{ii}=1$ and $\delta_{ii'}=0$ for $i\neq i'$.
The QFT over the symmetric group can be efficiently implemented with the circuit complexity $O(\mathrm{poly}(n))$~\cite{beals1997quantum,kawano2016quantum}.
Since
\begin{align}
    \mathrm{QFT}_{\mfS_n}^\dagger \ket{(n),1,1} =  \ket{+_{\mfS_n}}
\end{align}
holds, where $(n)$ corresponds to the trivial representation $g_{(n)}(\sigma)\coloneqq 1$ for all $\sigma\in\mfS_n$, the quantum state $\ket{+_{\mfS_n}}$ can also be efficiently implemented.

\subsection{Quantum superchannels}
\label{subsec:quantum_superchannels}

A quantum superchannel is a linear map that transforms quantum channels to quantum channels.
It is given as a linear map
\begin{align}
  \mcC: \bigotimes_{i=1}^{k-1} [\mcL(\mcO_i) \to \mcL(\mcI_{i+1})] \to [\mcL(\mcI_1) \to \mcL(\mcO_k)],
\end{align}
which transforms quantum channels $\Phi_i: \mcL(\mcO_i) \to \mcL(\mcI_{i+1})$ for $i\in [k-1]$ into a quantum channel $\Phi_\mathrm{out}: \mcL(\mcI_1) \to \mcL(\mcO_k)$, i.e., $\Phi_\mathrm{out} = \mcC(\Phi_1\otimes \cdots \otimes \Phi_{k-1})$.
Here, $\mcL(\mcH)$ represents the set of linear maps on a Hilbert space $\mcH$, and $[\mcL(\mcH_1) \to \mcL(\mcH_2)]$ represents the set of linear maps from $\mcL(\mcH_1)$ to $\mcL(\mcH_2)$.

We use the Choi--Jamio\l{}kowski isomorphism~\cite{CHOI1975285, JAMIOLKOWSKI1972275} to represent the quantum channels and superchannels.
A quantum channel $\Phi: \mcL(\mcI) \to \mcL(\mcO)$ is represented as its Choi matrix defined by
\begin{align}
  J_\Phi \coloneqq \sum_{i,j=1}^{d_{\mcI}} \ketbra{i}{j}_{\mcI} \otimes \Phi(\ketbra{i}{j})_{\mcO} \in \mcL(\mcI \otimes \mcO),
\end{align}
using the computational basis $\{\ket{i}\}_{i=1}^{d_{\mcI}}$ of $\mcI$, where $d_{\mcI}$ is the dimension of $\mcI$.
The Choi matrix of an isometry channel $\mcV(\cdot)\coloneqq V(\cdot)V^\dagger$ for an isometry operator $V: \mcI \to \mcO$ is given by
\begin{align}
  J_{\mcV} = \dketbra{V},
\end{align}
where $\dket{V}$ is the Choi vector of $V$ defined by $\dket{V} \coloneqq \sum_{i=1}^{d_{\mcI}} \ket{i}_{\mcI} \otimes V\ket{i}_{\mcO}$.
The composition of quantum channels $\Phi_1: \mcL(\mcX) \to \mcL(\mcY)$ and $\Phi_2: \mcL(\mcY) \to \mcL(\mcZ)$ is represented by the link product~\cite{chiribella2008quantum} of their Choi matrices defined by
\begin{align}
  J_{\Phi_2 \circ \Phi_1}
  &= J_{\Phi_2} \ast J_{\Phi_1}\\
  &\coloneqq \Tr_{\mcY}\left[
    (\1_{\mcX} \otimes J_{\Phi_2}) (J_{\Phi_1}^{\mathsf{T}_{\mcY}} \otimes \1_{\mcZ})\right],
\end{align}
where $\ast$ denotes the link product, and $(\cdot)^{\mathsf{T}_{\mcY}}$ denotes the partial transpose with respect to the subsystem $\mcY$.
The Choi matrix of a quantum superchannel $\mcC$ is defined by $J_{\mcC} \in \mcL(\bigotimes_{i=1}^{k} (\mcI_i \otimes \mcO_i))$ such that
\begin{align}
  J_{\mcC}\ast (J_{\Phi_1} \otimes \cdots \otimes J_{\Phi_{k-1}}) = J_{\Phi_\mathrm{out}}
\end{align}
for all linear maps $\Phi_i: \mcL(\mcO_i) \to \mcL(\mcI_{i+1})$ for $i\in [k-1]$ and $\Phi_\mathrm{out} \coloneqq \mcC(\Phi_1\otimes \cdots \otimes \Phi_{k-1})$.

The quantum comb is a quantum superchannel implementable on the quantum circuit model with a fixed causal order~\cite{chiribella2008quantum}, which is given by a sequential connection of quantum channels $\Lambda_i: \mcL(\mcI_i \otimes \mcA_{i-1}) \to \mcL(\mcO_i \otimes \mcA_i)$ for $i\in [k]$ with intermediate memory systems $\mcA_i$ satisfying $\mcA_0\simeq \mcA_k\simeq \CC$, namely,
\begin{align}
  &\mcC(\Phi_1\otimes \cdots \otimes \Phi_{k-1})\nonumber\\
  &= \Lambda_k \circ [\Phi_{k-1} \otimes \1_{\mcA_{k-1}}] \circ \Lambda_{k-1} \circ \cdots \circ [\Phi_1\otimes \1_{\mcA_1}] \circ \Lambda_1.
\end{align}
The Choi matrix of a quantum comb $\mcC$ is given by $J_{\mcC} = J_{\Lambda_1} \ast \cdots \ast J_{\Lambda_k}$.

In the quantum circuit model, we can consider a more general class of quantum superchannels, called the quantum circuit with classical control of causal order, or a QC-CC superchannel~\cite{wechs2021quantum}.
It allows the causal order of the queries to be classically controlled by the measurement outcomes of the intermediate memory systems.
It is believed to be the most general class of superchannels that can be implemented by quantum circuits with a fixed number of queries to input channels.
It is given by a sequential connection of quantum instruments\footnote{The quantum instrument $\{\Lambda_{a_1\ldots a_{i-1}}^{(a_i)}\}_{a_i}$ is a set of linear maps such that $\Lambda_{a_1\ldots a_{i-1}}^{(a_i)}$ is completely positive (CP) for all $a_i$ and $\sum_{a_i} \Lambda_{a_1\ldots a_{i-1}}^{(a_i)}$ is trace preserving (TP).} $\{\Lambda_{a_1\ldots a_{i-1}}^{(a_i)}: \mcL(\mcI_i \otimes \mcA_{i-1}) \to \mcL(\mcO_i \otimes \mcA_i)\}_{a_i}$ for $i\in [k]$, where $a_i$ represents the measurement outcome of $i$-th quantum instrument.
The measurement outcome $a_i$ determines the quantum channel to be used in the $i$-th slot, namely,
\begin{align}
  &\mcC(\Phi_1\otimes \cdots \otimes \Phi_{k-1})\nonumber\\
  &= \sum_{a_1, \ldots, a_{k-1}}\Lambda_{a_1\ldots a_{k-1}} \circ [\Phi_{a_{k-1}} \otimes \1_{\mcA_{k-1}}] \circ \Lambda^{(a_{k-1})}_{a_1\ldots a_{k-2}} \circ \nonumber\\
  &\hspace{60pt}\cdots \circ [\Phi_{a_1}\otimes \1_{\mcA_1}] \circ \Lambda^{(a_1)}_{\varnothing}
\end{align}
for all linear maps $\Phi_i: \mcL(\mcO_i) \to \mcL(\mcI_{i+1})$ for $i\in [k-1]$ satisfying $\mcI_1\simeq \cdots \simeq \mcI_{k-1}$ and $\mcO_1\simeq \cdots \simeq \mcO_{k-1}$, where $\varnothing$ denotes the empty set.
Note that the last quantum channel $\Lambda_{a_1\ldots a_{k-1}}$ does not have the measurement outcome since it is not followed by any query.

\subsection{Dilation of quantum channels and quantum combs}
\label{subsec:dilation_of_channels_and_combs}

Any quantum state $\rho\in \mcL(\mcH)$ has its purification $\ket{\psi}\in \mcE \otimes \mcH$ such that $\rho = \Tr_{\mcE}[\ketbra{\psi}]$.
If the rank of $\rho$ is $r$, the dimension of the environment $\mcE$ can be taken as $r$, i.e., $\mcE\simeq \CC^r$.
Purification of $\rho$ is not unique, and we define the set of purifications of $\rho$ as
\begin{align}
  \mathrm{Pur}_r(\rho) \coloneqq \{\ket{\psi} \in \mcE \otimes \mcH \mid \rho = \Tr_{\mcE}[\ketbra{\psi}]\}.
\end{align}
Once a purification $\ket{\psi_0}$ of $\rho$ is fixed, the set of purifications of $\rho$ can be represented as
\begin{align}
  \mathrm{Pur}_r(\rho) = \{(U \otimes \1_{\mcH})\ket{\psi_0} \mid U \in \U(r)\},
\end{align}
and we define the uniform distribution on $\mathrm{Pur}_r(\rho)$ as the distribution induced by the Haar measure on $\U(r)$.

Similarly to the purification of quantum states, any quantum channel $\Lambda: \mcL(\mcI) \to \mcL(\mcO)$ has its dilation given by an isometry $V: \mcI \to \mcE \otimes \mcO$ such that $\Lambda(\cdot) = \Tr_{\mcE}[V (\cdot) V^\dagger]$.
If the Kraus rank of $\Lambda$, defined by the rank of $J_\Lambda$, is $r$, the dimension of the environment $\mcE$ can be taken as $r$, i.e., $\mcE\simeq \CC^r$.
The set of dilations of $\Lambda$ is defined as
\begin{align}
  \mathrm{Dil}_r(\Lambda) \coloneqq \{V: \mcI \to \mcE \otimes \mcO \mid \Lambda(\cdot) = \Tr_{\mcE}[V (\cdot) V^\dagger]\}.
\end{align}
The sets of purifications and dilations are related by the Choi--Jamio\l{}kowski isomorphism as
\begin{align}
  \label{eq:dilation_purification_relation}
  \mathrm{Dil}_r(\Lambda) = \{\mcV \mid \dket{V}\in \mathrm{Pur}_r(J_\Lambda)\}.
\end{align}
We define the uniform distribution on $\mathrm{Dil}_r(\Lambda)$ as the distribution induced by the uniform distribution on $\mathrm{Pur}_r(J_\Lambda)$.

Any quantum superchannel $\mcC$ has its dilation given by a pure quantum comb~\cite{araujo2017purification,yokojima2021consequences} $\mcC_\mathrm{pur}$ such that $\mcC(\cdot) = \Tr_{\mcE}[\mcC_\mathrm{pur}(\cdot)]$.
A pure quantum comb is defined as a quantum comb composed of isometry operators $V_i: \mcI_i \otimes \mcB_{i-1} \to \mcO_i \otimes \mcB_i$ for $i\in [k]$ with intermediate memory systems $\mcB_i$ satisfying $\mcB_0\simeq \CC$ and $\mcB_k = \mcE$, namely,
\begin{align}
  &\mcC_\mathrm{pur}(\Phi_1\otimes \cdots \otimes \Phi_{k-1})\nonumber\\
  &= \mcV_k \circ [\Phi_{k-1} \otimes \1_{\mcB_{k-1}}] \circ \mcV_{k-1} \circ \cdots \circ [\Phi_1\otimes \1_{\mcB_1}] \circ \mcV_1.
\end{align}
If the rank of $J_\mcC$ is $r$, the dimension of the environment $\mcE$ can be taken as $r$, and we can define the set of dilations of $\mcC$ as
\begin{align}
  \mathrm{Dil}_r(\mcC) \coloneqq \{\mcC_\mathrm{pur}: \text{pure} \mid \mcC(\cdot) = \Tr_{\mcE}[\mcC_\mathrm{pur}(\cdot)]\}.
\end{align}
The set of dilations of a quantum superchannel $\mcC$ is related to the set of purifications of $J_\mcC$ as
\begin{align}
  \mathrm{Dil}_r(\mcC) = \{\mcC_\mathrm{pur} \mid J_{\mcC_\mathrm{pur}} = \ketbra{\psi} \text{ for }\ket{\psi} \in \mathrm{Pur}_r(J_\mcC)\}.
\end{align}
We define the uniform distribution on $\mathrm{Dil}_r(\mcC)$ as the distribution induced by the uniform distribution on $\mathrm{Pur}_r(J_\mcC)$.

\subsection{No-go theorem on simulating the quantum switch by a QC-CC superchannel}
\label{subsec:switch_simulation}

There exists a quantum superchannel that cannot be implemented in the quantum circuit model, such as the quantum switch~\cite{chiribella2013quantum}, which transforms two quantum channels $\Lambda$ and $\Phi$ into a quantum channel that applies $\Lambda$ and $\Phi$ in a superposition of two different orders.
It is defined as a $2$-slot quantum superchannel $\mcS_\mathrm{switch}$ given by
\begin{align}
  \mcS_\mathrm{switch}(\Lambda \otimes \Phi)(\cdot) \coloneqq \sum_{i,j} S_{ij} (\cdot) S_{ij}^\dagger,\label{eq:def_quantum_switch}\\
  S_{ij} \coloneqq \ketbra{0}{0}\otimes K_i L_j  + \ketbra{1}{1}\otimes  L_j K_i,
\end{align}
where  $K_i$ and $L_j$ are the Kraus operators of $\Lambda$ and $\Phi$ respectively, i.e., $\Lambda(\cdot) = \sum_i K_i (\cdot) K_i^\dagger$ and $\Phi(\cdot) = \sum_j L_j (\cdot) L_j^\dagger$.

The quantum switch cannot be implemented by a QC-CC superchannel if we just have single queries of $\Lambda$ and $\Phi$~\cite{wechs2021quantum}.
In addition, even if we have $n$ queries of $\Lambda$ and a single query of $\Phi$, the implementation of the quantum channel $\mcS_\mathrm{switch}(\Lambda\otimes \Phi)$ by a QC-CC superchannel requires exponentially many $n$, as shown in Ref.~\cite{bavaresco2025simulating}.
\begin{theorem}[\cite{bavaresco2025simulating}]
  \label{thm:quantum_switch_no-go}
  There is no QC-CC superchannel $\mcC$ such that $\mcC(\Lambda^{\otimes n} \otimes \Phi) = \mcS_\mathrm{switch}(\Lambda \otimes \Phi)$ for all $N$-qubit quantum channels $\Lambda, \Phi$ if
  \begin{align}
    n\leq \max(2, 2^N-1).
  \end{align}
\end{theorem}

\section{Main result 1: Parallel random dilation superchannel}
\label{sec:parallel_random_dilation_superchannel}

\subsection{Random purification channel}
\label{subsec:random_purification_channel}

\begin{figure}
  \centering
  \includegraphics[width=\linewidth]{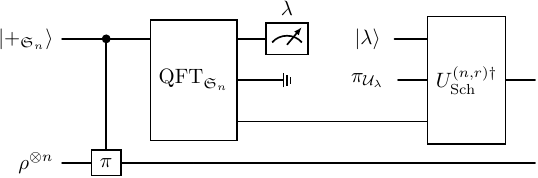}
  \caption{Quantum circuit implementing the random purification channel $\Phi$ that transforms $n$ parallel queries of an unknown quantum state $\rho \in \mcL(\CC^d)$ into $n$ parallel queries of a randomly chosen purification of $\rho$.
  See Eqs.~\eqref{eq:schur_weyl_duality}--\eqref{eq:def_pi_u} for the definitions of the quantum Schur transform $U_\mathrm{Sch}$, the controlled permutation unitary $\mathrm{ctrl}-\pi$, the quantum Fourier transform over the symmetric group $\mathrm{QFT}_{\mfS_n}$, the quantum state $\ket{+_{\mfS_n}}$, and the maximally mixed state $\pi_{\mcU_{\lambda}}$.}
  \label{fig:random-purification-channel}
\end{figure}

We first present a quantum circuit implementing the random purification channel based on the quantum Fourier transform over the symmetric group and the quantum Schur transform, as shown in Fig.~\ref{fig:random-purification-channel}.
This construction is different from the original construction in Ref.~\cite{tang2025conjugate}, and it is more suitable for our later construction of the random dilation superchannel.

\begin{theorem}
  \label{thm:random-purification-channel}
  The quantum circuit shown in Fig.~\ref{fig:random-purification-channel} implements the random purification channel $\Phi$ that transforms $n$ parallel queries of an unknown quantum state $\rho \in \mcL(\CC^d)$ with rank at most $r$ into $n$ parallel queries of a randomly chosen purification of $\rho$, i.e.,
  \begin{align}
    \Phi\left(\rho^{\otimes n}\right)
    =
    \mathbb{E}_{\ket{\psi} \sim \mathrm{Pur}_r(\rho)}
    \left[
      \ketbra{\psi}^{\otimes n}
    \right],
  \end{align}
  where the expectation is taken over the uniform distribution on $\mathrm{Pur}_r(\rho)$.
  The circuit complexity of implementing the quantum channel $\Phi$ is $O(\mathrm{poly}(n, \log d))$.
\end{theorem}

\begin{proof}
  As shown in Lem. 2.16 of Ref.~\cite{tang2025conjugate}, the random purification is given by
  \begin{align}
    &\mathbb{E}_{\ket{\psi} \sim \mathrm{Pur}_r(\rho)}
    \left[
      \ketbra{\psi}^{\otimes n}
    \right]\nonumber\\
    &= \sum_{\lambda \vdash n} \sum_{i, i'=1}^{m_\lambda} U_\mathrm{Sch}^{(n,r)\dagger} (\ketbra{\lambda} \otimes \pi_{\mcU_{\lambda}}\otimes \ketbra{i}{i'}) U_\mathrm{Sch}^{(n,r)} \nonumber\\
    &\quad \otimes U_\mathrm{Sch}^{(n,d)\dagger} \left[\ketbra{\lambda} \otimes f_{\lambda}(\rho) \otimes \ketbra{i}{i'}\right] U_\mathrm{Sch}^{(n,d)}.
  \end{align}
The quantum state just after the QFT in the quantum circuit shown in Fig.~\ref{fig:random-purification-channel} is given by
\begin{align}
  &{1\over n!}\sum_{\sigma, \tau\in \mfS_n} \mathrm{QFT}_{\mfS_n} \ketbra{\sigma}{\tau} \mathrm{QFT}_{\mfS_n}^\dagger \otimes \pi(\sigma)\rho^{\otimes n} \pi(\tau)^\dagger \nonumber\\
  &= {1\over n!}\sum_{\sigma, \tau\in \mfS_n}\sum_{\lambda, \lambda'\vdash n}\sum_{i,j=1}^{m_\lambda}\sum_{i', j'=1}^{m_{\lambda'}} {\sqrt{m_\lambda m_{\lambda'}} \over n!} [g_{\lambda}(\sigma)]_{ji} [\overline{g_{\lambda'}(\tau)}]_{j'i'} \nonumber\\
  &\quad \times \ketbra{\lambda, i, j}{\lambda', i', j'} \nonumber\\
  &\quad \otimes U_\mathrm{Sch}^{(n,d)\dagger}\left[\sum_{\mu\vdash n} \ketbra{\mu} \otimes f_\mu(\rho) \otimes g_{\mu}(\sigma) g_{\mu}(\tau)^\dagger\right]U_\mathrm{Sch}^{(n,d)}.
\end{align}
Thus, the output quantum state is given by
\begin{align}
  &{1\over n!}\sum_{\sigma, \tau\in \mfS_n}\sum_{\lambda\vdash n}\sum_{i, j, j'=1}^{m_\lambda} {m_{\lambda} \over n!} [g_{\lambda}(\sigma)]_{ji} [\overline{g_{\lambda}(\tau)}]_{j'i} \nonumber\\
  &\quad \times U_\mathrm{Sch}^{(n,r)\dagger} (\ketbra{\lambda} \otimes \pi_{\mcU_\lambda} \otimes \ketbra{j}{j'}) U_\mathrm{Sch}^{(n,r)} \nonumber\\
  &\quad \otimes U_\mathrm{Sch}^{(n,d)\dagger}\left[\sum_{\mu\vdash n} \ketbra{\mu} \otimes f_\mu(\rho) \otimes g_{\mu}(\sigma) g_{\mu}(\tau)^\dagger\right]U_\mathrm{Sch}^{(n,d)} \nonumber\\
  &= {1\over n!}\sum_{\sigma, \tau\in \mfS_n}\sum_{\lambda, \mu\vdash n}\sum_{i, j, j'=1}^{m_\lambda} \sum_{k,l,k,l'=1}^{m_\mu}\nonumber\\
  &\qquad {m_{\lambda} \over n!} [g_{\lambda}(\sigma)]_{ji} [\overline{g_{\lambda}(\tau)}]_{j'i}[g_{\mu}(\sigma)]_{kl} [\overline{g_{\mu}(\tau)}]_{k'l'} \nonumber \\
  &\quad \times U_\mathrm{Sch}^{(n,r)\dagger} (\ketbra{\lambda} \otimes \pi_{\mcU_\lambda} \otimes \ketbra{j}{j'}) U_\mathrm{Sch}^{(n,r)} \nonumber\\
  &\quad \otimes U_\mathrm{Sch}^{(n,d)\dagger}\left[\ketbra{\mu} \otimes f_\mu(\rho) \otimes \ketbra{k}{l}\cdot\ketbra{l'}{k'} \right]U_\mathrm{Sch}^{(n,d)}\\
  &= \sum_{\lambda\vdash n}\sum_{i, j, j'=1}^{m_\lambda} {1\over m_\lambda} U_\mathrm{Sch}^{(n,r)\dagger} (\ketbra{\lambda} \otimes \pi_{\mcU_\lambda} \otimes \ketbra{j}{j'}) U_\mathrm{Sch}^{(n,r)} \nonumber\\
  &\quad \otimes U_\mathrm{Sch}^{(n,d)\dagger}\left[\ketbra{\lambda} \otimes  f_\lambda(\rho) \otimes \ketbra{j}{i} \cdot \ketbra{i}{j'}\right]U_\mathrm{Sch}^{(n,d)}\\
  &= \sum_{\lambda\vdash n}\sum_{j, j'=1}^{m_\lambda} U_\mathrm{Sch}^{(n,r)\dagger} (\ketbra{\lambda} \otimes \pi_{\mcU_\lambda} \otimes \ketbra{j}{j'}) U_\mathrm{Sch}^{(n,r)} \nonumber\\
  &\quad \otimes U_\mathrm{Sch}^{(n,d)\dagger}\left[\ketbra{\lambda} \otimes  f_\lambda(\rho) \otimes \ketbra{j}{j'}\right]U_\mathrm{Sch}^{(n,d)}\\
  &= \mathbb{E}_{\ket{\psi} \sim \mathrm{Pur}_r(\rho)}
    \left[
      \ketbra{\psi}^{\otimes n}
    \right],
\end{align}
where we use the Schur orthogonality relation~\eqref{eq:schur_orthogonality} and the reality of $g_{\lambda}(\sigma)$ in the Young--Yamanouchi basis~\cite{james2006representation} in the second equality.
\end{proof}

Note that if we trace out the environment system $\mcE$, we obtain the generalized phase estimation~\cite{harrow2005applications} on the input state $\ket{\psi_{\mathrm{in}}}$, which is to output the measurement outcome $\lambda\vdash n$ with the post measurement state given by $\Pi_\lambda \ket{\psi_{\mathrm{in}}}\bra{\psi_{\mathrm{in}}} \Pi_\lambda / \Tr(\Pi_\lambda \ketbra{\psi_{\mathrm{in}}})$, where $\Pi_\lambda$ is the Young projector defined by $\Pi_\lambda\coloneqq U_\mathrm{Sch}^{\dagger (n,d)} (\1_{\mcU_\lambda} \otimes \1_{\mcS_\lambda}) U_\mathrm{Sch}^{(n,d)}$.

\subsection{Parallel random dilation superchannel and supersuperchannel}
\label{subsec:parallel_random_dilation_superchannel}

By using the random purification channel shown in Fig.~\ref{fig:random-purification-channel}, we present the main result of this work, which is to construct the quantum circuit implementing the random dilation of an unknown quantum channel.

\begin{figure*}
  \centering
  \includegraphics{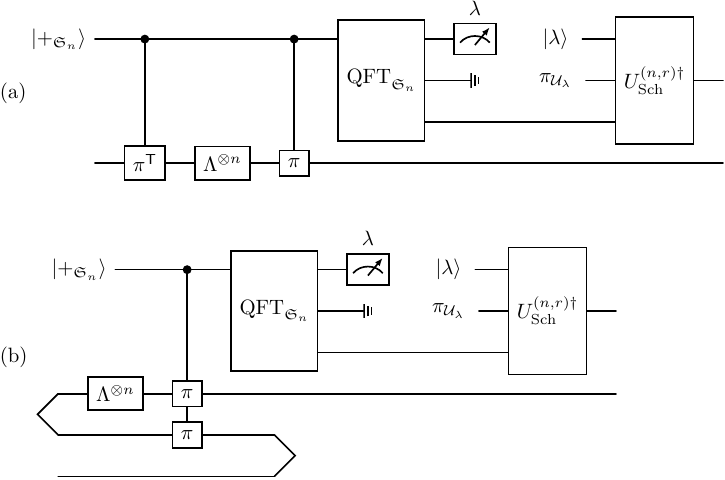}
  \caption{(a) Quantum circuit implementing the random dilation superchannel $\Xi$ that transforms $n$ parallel queries of an unknown quantum channel $\Lambda:\mcL(\mcI) \to \mcL(\mcO)$ into $n$ parallel queries of a randomly chosen dilation isometry of $\Lambda$.
  (b) A proof idea of implementing the random dilation superchannel $\Xi$ using the closed timelike curve based on postselected teleportation (P-CTC)~\cite{lloyd2011quantum}.
  The symbol ``$<$'' represents the unnormalized maximally entangled state $\dket{\1}$, and the symbol ``$>$'' represents the P-CTC.
  By replacing the controlled permutation channel $\mathrm{ctrl}-\pi$ sandwiched by $\dket{\1}$ and the P-CTC with the transposed controlled permutation channel $\mathrm{ctrl}-\pi^{\mathsf{T}}$, we obtain the quantum circuit in (a).}
  \label{fig:random-dilation-superchannel}
\end{figure*}

\begin{theorem}
  \label{thm:random-dilation-superchannel}
  The quantum circuit shown in Fig.~\ref{fig:random-dilation-superchannel}~(a) implements the random dilation superchannel $\Xi$ that transforms $n$ parallel queries of an unknown quantum channel $\Lambda:\mcL(\mcI) \to \mcL(\mcO)$ with at most Kraus rank $r$ into $n$ parallel queries of a randomly chosen dilation isometry of $\Lambda$, i.e.,
  \begin{align}
    \Xi\left(\Lambda^{\otimes n}\right)
    =
    \mathbb{E}_{V\sim \mathrm{Dil}_r(\Lambda)}
    \left[
      \mathcal{V}^{\otimes n}
    \right],
  \end{align}
  where the expectation is taken over the uniform distribution on $\mathrm{Dil}_r(\Lambda)$.
  The circuit complexity of implementing the quantum superchannel $\Xi$ is $O(\mathrm{poly}(n, \log d_{\mcI}, \log d_{\mcO}))$, where $d_{\mcO} = \dim(\mcO)$.
\end{theorem}
\begin{proof}
  We show the proof based on the one-to-one correspondence between the purification of the Choi matrix $J_\Lambda$ and the dilation of $\Lambda$ shown in Eq.~\eqref{eq:dilation_purification_relation}.
  If we allow an unnormalized maximally entangled state $\dket{\1}$ and the closed timelike curve based on postselected teleportation (P-CTC)~\cite{lloyd2011quantum} in the quantum circuit, the random dilation isometry can be implemented in the following way\footnote{Note that this strategy does not physically implement store-and-retrieve of quantum channels. It illustrates the proof idea of combining the ``transpose trick'' and the Choi–Jamiołkowski isomorphism between the set of dilation isometries $\mathrm{Dil}_r(\Lambda)$ and the set of purifications $\mathrm{Pur}_r(J_{\Lambda})$.}:
  \begin{enumerate}
    \item Store the action of $\Lambda$ into the Choi state $J_{\Lambda}$.
    \item Apply the random purification channel on $J_{\Lambda}^{\otimes n}$ to obtain the Choi state of $\mathbb{E}_{V\sim \mathrm{Dil}_r(\Lambda)}[\mcV^{\otimes n}]$.
    \item Retrieve the action of the quantum channel $\mathbb{E}_{V\sim \mathrm{Dil}_r(\Lambda)}[\mcV^{\otimes n}]$ from its Choi state by using the P-CTC.
  \end{enumerate}
  The above strategy can be described as Fig.~\ref{fig:random-dilation-superchannel}~(b).
  The controlled permutation channel $\mathrm{ctrl}-\pi$ sandwiched by the unnormalized maximally entangled state $\dket{\1}$ and the P-CTC can be rewritten as
  \begin{align}
    \mathrm{ctrl}-\pi^{\mathsf{T}} \coloneqq \sum_{\sigma \in \mfS_n} \ketbra{\sigma}{\sigma} \otimes \pi(\sigma)^{\mathsf{T}},
  \end{align}
  as shown in Fig.~\ref{fig:random-dilation-superchannel}~(a).
\end{proof}

The above strategy can be straightforwardly extended to the case of quantum superchannels (random dilation supersuperchannel).
We can construct the random dilation supersuperchannel $\Sigma$ that transforms $n$ parallel queries of an unknown quantum superchannel $\mcC$ into $n$ parallel queries of a randomly chosen dilation superchannel of $\mcC$ as follows.

\begin{figure*}
  \centering
  \includegraphics[width=\linewidth]{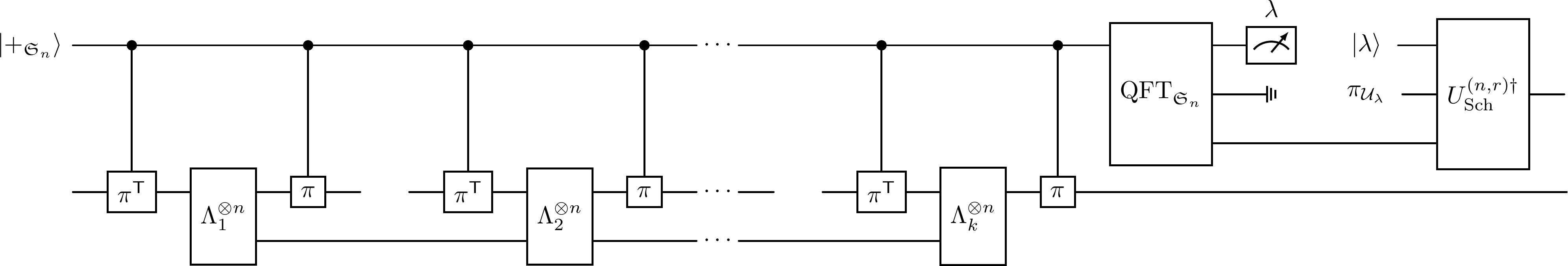}
  \caption{Quantum circuit implementing the random dilation supersuperchannel $\Sigma$ that transforms $n$ parallel queries of an unknown quantum superchannel $\mcC$ into $n$ parallel queries of a randomly chosen dilation superchannel of $\mcC$.}
  \label{fig:superchannel_random_purification}
\end{figure*}

\begin{theorem}
  \label{thm:random-purification-supersuperchannel}
  The quantum circuit shown in Fig.~\ref{fig:superchannel_random_purification} implements the random dilation supersuperchannel $\Sigma$ that transforms $n$ parallel queries of an unknown quantum superchannel $\mcC$ with at most rank $r$ into $n$ parallel queries of a randomly chosen dilation superchannel of $\mcC$, i.e.,
  \begin{align}
    \Sigma\left(\mcC^{\otimes n}\right)
    =
    \mathbb{E}_{\mcC_\mathrm{pur} \sim \mathrm{Dil}_r(\mcC)}
    \left[
      \mcC_\mathrm{pur}^{\otimes n}
    \right],
  \end{align}
  where the expectation is taken over the uniform distribution on $\mathrm{Dil}_r(\mcC)$.
  The circuit complexity of implementing the quantum superchannel $\Sigma$ is $O(\mathrm{poly}(n, k, \log d_{I_i}, \log d_{O_i}))$, where $d_{I_i} = \dim(\mcI_i)$ and $d_{O_i} = \dim(\mcO_i)$.
\end{theorem}

\subsection{Another construction of parallel random dilation superchannel}
\label{sec:another_construction}

\begin{figure*}
  \centering
  \includegraphics[width=\linewidth]{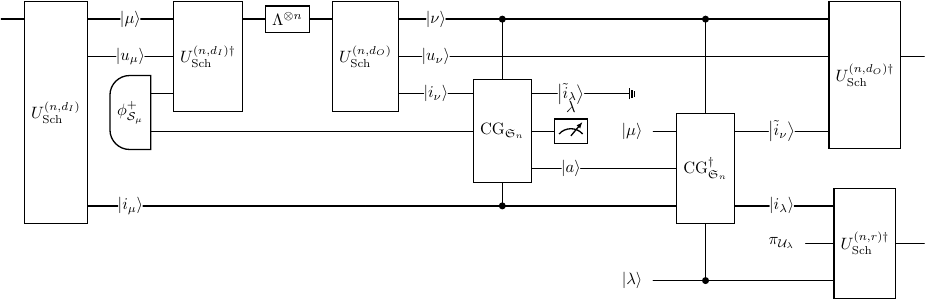}
  \caption{Quantum circuit implementing another construction of the random dilation superchannel $\mcC$ that transforms $n$ parallel queries of an unknown quantum channel $\Lambda:\mcL(\mcI) \to \mcL(\mcO)$ into $n$ parallel queries of a randomly chosen dilation isometry of $\Lambda$.
  The quantum gate $\mathrm{CG}_{\mathfrak{S}_n}$ represents the Kronecker transform with irreps specified by the controlled qubits.
  The control wire from the state $\ket{i_\mu}$ represents the irrep $\mu$.
  The quantum state $\ket{\phi^+_{\mcS_\mu}}$ is the maximally entangled state $\ket{\phi^+_{\mcS_\mu}} \coloneqq {1\over \sqrt{m_\mu}}\sum_{i} \ket{i}\otimes \ket{i}$ using the Young--Yamanouchi basis $\{\ket{i}\}_{i=1}^{m_\mu}$ of $\mcS_\mu$.
  }
  \label{fig:channel_random_dilation_kronecker}
\end{figure*}

We present another construction of the random dilation superchannel based on the quantum Schur transform and the Clebsch--Gordan transform over the symmetric group, which we call the \emph{Kronecker transform}, as shown in Fig.~\ref{fig:channel_random_dilation_kronecker}.
The Kronecker transform is defined as follows.
Given two irreps $g_\mu, g_\nu$ of the symmetric group $\mathfrak{S}_n$ for $\mu, \nu\vdash n$, the tensor product representation $g_\mu \otimes g_\nu$ can be decomposed into irreps as
\begin{align}
  g_\mu(\sigma) \otimes g_\nu(\sigma) \simeq \bigoplus_{\lambda \vdash n} g_\lambda(\sigma) \otimes \1_{g_{\mu\nu\lambda}} \quad \forall \sigma\in \mathfrak{S}_n,
\end{align}
where $g_{\mu\nu\lambda}$ is the multiplicity of $g_\lambda$ appearing in the decomposition of $g_\mu \otimes g_\nu$, which is called the Kronecker coefficient.
This decomposition defines an isomorphism between the spaces
\begin{align}
  \mcS_\mu\otimes \mcS_\nu \simeq \bigoplus_{\lambda \vdash n} \mcS_\lambda \otimes \CC^{g_{\mu\nu\lambda}},
\end{align}
and we define the Kronecker transform $\mathrm{CG}_{\mu\nu}$ that represents the basis change corresponding to this isomorphism, where we take an orthonormal basis $\{\ket{a}\}$ of the multiplicity space $\CC^{g_{\mu\nu\lambda}}$.
In the quantum circuit, the Kronecker transform is represented as $\mathrm{CG}_{\mathfrak{S}_n}$ with the controlled qubits in the states $\ket{\mu}$ and $\ket{\nu}$.
We also define the maximally entangled state $\ket{\phi^+_{\mcS_\mu}} \coloneqq {1\over \sqrt{m_\nu}}\sum_{i} \ket{i}\otimes \ket{i}$ using the Young--Yamanouchi basis $\{\ket{i}\}_{i=1}^{m_\mu}$ of $\mcS_\mu$.
We can show that the quantum circuit shown in Fig.~\ref{fig:channel_random_dilation_kronecker} implements the random dilation superchannel (see Appendix~\ref{appendix:another_construction_random_dilation_superchannel} for the proof).

\subsection{Approximate lifting of parallel dilation isometries to sequential dilation unitaries}
\label{subsec:lifting_parallel_to_sequential}

We present the lifting of parallel dilation isometries to sequential dilation unitaries.
The special case ($d=1$) of this theorem reproduces the sample-to-query lifting shown in Ref.~\cite{tang2025conjugate}.
\begin{theorem}
  \label{thm:lifting_parallel_to_sequential}
  For any isometry $V: \CC^d \to \CC^D$ and a fixed isometry $V_0: \CC^d \to \CC^D$, the quantum channel $\mcV^{\otimes n}$ for $n = O(m^2d^2/\varepsilon)$ can be converted into $m$ sequential queries of random unitary $U_\theta$ on $\mcH \coloneqq \CC^2 \otimes \CC^D$ such that $U_\theta[\ket{0} \otimes V_0 (\cdot)] = e^{-i\theta}\ket{1} \otimes V(\cdot)$ with approximation error $\varepsilon$ in the diamond norm, where $\theta\in [0,2\pi)$ is randomly chosen.
\end{theorem}

To this end, we first construct a quantum circuit converting $n$ parallel queries of the isometry channel $\mcV$ given by $\mcV^{\otimes n}$ into $n$ parallel queries of conditional isometry $V(\theta)$ (see Fig.~\ref{fig:conditional_isometry}~(a)):
\begin{align}
  \label{eq:parallel_conditional_isometry}
  \mathbb{E}_{\theta\sim [0, 2\pi)} \left[
    V(\theta)^{\otimes n}(\cdot)V(\theta)^{\dagger\otimes n}
  \right], 
\end{align}
where $\theta$ is randomly chosen from the uniform distribution on $[0, 2\pi)$, and $V(\theta): \CC^d \to \mcH$ is defined by
\begin{align}
  \label{eq:conditional_isometry}
  V(\theta) \coloneqq {1\over \sqrt{2}} (e^{i\theta}\ket{0} \otimes V_0 + \ket{1} \otimes V),
\end{align}
using a fixed isometry operator $V_0: \CC^d \to \CC^D$.

\begin{figure*}
  \centering
  \includegraphics[width=0.7\linewidth]{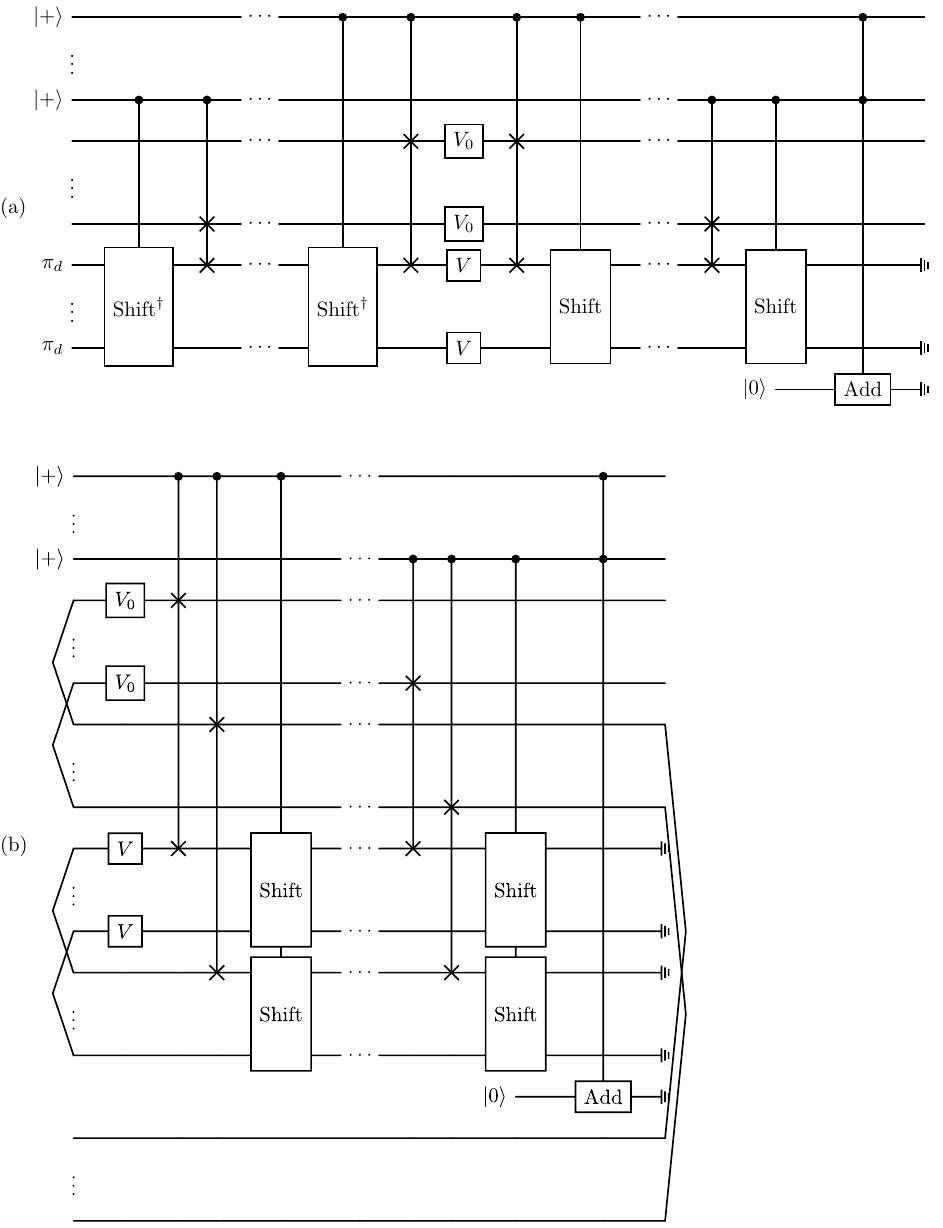}
  \caption{(a) Quantum circuit implementing $n$ parallel queries of the conditional isometry $V(\theta)$ defined in Eq.~\eqref{eq:conditional_isometry}, where $\mathrm{Shift}$ and $\mathrm{Add}$ are unitary operators defined in Eqs.~\eqref{eq:shift_operator} and \eqref{eq:add_operator}, and $\pi_d$ is the maximally mixed state given by $\pi_d\coloneqq \1_d/d$.
  (b) A proof idea of implementing the quantum channel given in Eq.~\eqref{eq:parallel_conditional_isometry} using the P-CTC.
  Similarly to Fig.~\ref{fig:random-dilation-superchannel}~(b), the controlled SWAP, $V_0$, $\mathrm{Shift}$ and the trace out operations sandwiched by the maximally entangled states and P-CTCs in Fig.~\ref{fig:conditional_isometry}~(b) can be rewritten as the controlled SWAP, $V_0$, $\mathrm{Shift}^\dagger$, and the maximally mixed state, respectively, to obtain the quantum circuit in (a).}
  \label{fig:conditional_isometry}
\end{figure*}
\begin{lemma}
  \label{lemma:parallel_to_sequential_conditional_isometry}
  For any isometry $V: \CC^d \to \CC^D$ and a fixed isometry $V_0: \CC^d \to \CC^D$, the quantum channel $\mcV^{\otimes n}$ can be converted into $n$ parallel queries of the random conditional isometry $V(\theta)$ using the quantum circuit shown in Fig.~\ref{fig:conditional_isometry}~(a).
\end{lemma}
\begin{proof}
We present the construction of this conversion using a similar idea to the random dilation superchannel, i.e., we first construct the Choi state transformation, retrieve the quantum channel transformation using the P-CTC, and then replace the P-CTC and the maximally entangled state with deterministic quantum channels.
To this end, we consider the conversion from samples to conditional samples shown in Ref.~\cite{tang2025conjugate}.
The task is, for a fixed state $\ket{\psi_0}\in \CC^d$, to convert $n$ copies of an unknown pure state $\ket{\psi}\in \CC^d$ into $n$ copies of the random conditional sample given by
\begin{align}
  \mathbb{E}_{\theta\in [0,2\pi)} \ketbra{\psi(\theta)}^{\otimes n},
\end{align}
where $\ket{\psi(\theta)}$ is defined by
\begin{align}
  \ket{\psi(\theta)} \coloneqq {1\over \sqrt{2}}(e^{i\theta} \ket{0}\otimes \ket{\psi_0} + \ket{1} \otimes \ket{\psi}).
\end{align}
This transformation is achieved by the quantum circuit shown in Fig.~\ref{fig:sample_to_conditional_sample}, which consists of the Shift operator defined by
\begin{align}
  \label{eq:shift_operator}
  \mathrm{Shift}\ket{i_1\ldots i_n} \coloneqq \ket{i_2 i_3 \ldots i_n i_1},
\end{align}
for all $i_1, \ldots, i_n \in [d]$, and the $\mathrm{Add}$ operator defined by
\begin{align}
  \label{eq:add_operator}
  \mathrm{Add}(\ket{i_1\ldots i_n} \otimes \ket{0}) \coloneqq \ket{i_1\ldots i_n} \otimes \ket{\sum_{k=1}^n i_k},
\end{align}
for all $i_1, \ldots, i_n \in \{0,1\}$.

\begin{figure}
  \centering
  \includegraphics{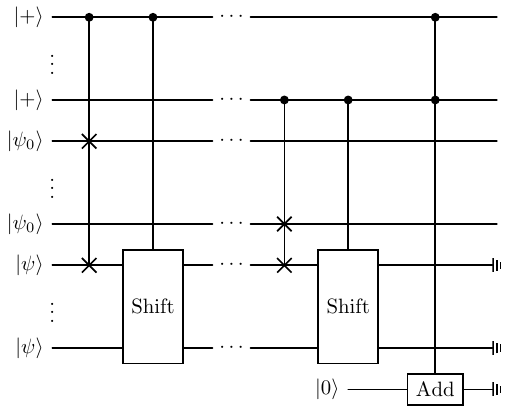}
  \caption{Quantum circuit implementing the conversion from $n$ copies of an unknown pure state $\ket{\psi}\in \CC^d$ into $n$ copies of the random conditional state given by $\mathbb{E}_{\theta\in [0,2\pi)} \ketbra{\psi(\theta)}^{\otimes n}$.}
  \label{fig:sample_to_conditional_sample}
\end{figure}

We consider the Choi state $\dket{V}$ corresponding to an isometry operator $V:\CC^d\to \CC^D$.
Then, the Choi state of the conditional isometry given in Eq.~\eqref{eq:parallel_conditional_isometry} is given by
\begin{align}
  \dket{V(\theta)} = {1\over \sqrt{2}} (e^{i\theta}\ket{0} \otimes \dket{V_0} + \ket{1}\otimes \dket{V}).
\end{align}
Thus, the Choi state
\begin{align}
  \mathbb{E}_{\theta\sim [0,2\pi)} \dketbra{V(\theta)}^{\otimes n}
\end{align}
corresponding to the Choi state of $n$ parallel queries of the random conditional isometry given in Eq.~\eqref{eq:parallel_conditional_isometry} is obtained by applying the conversion from samples to conditional samples to the Choi state $\dket{V}^{\otimes n}$.
By inserting the P-CTCs as shown in Fig.~\ref{fig:conditional_isometry}~(b), we can retrieve the quantum channel given in Eq.~\eqref{eq:parallel_conditional_isometry}.
The controlled SWAP, $V_0$, and the trace out operations sandwiched by the maximally entangled states and P-CTCs in Fig.~\ref{fig:conditional_isometry}~(b) can be rewritten as the controlled SWAP, $V_0$, and the maximally mixed state $\pi_d\coloneqq \1_d/d$, respectively, in the same way as the proof of Thm.~\ref{thm:random-dilation-superchannel}.
Then, we obtain the quantum circuit shown in Fig.~\ref{fig:conditional_isometry}~(a).
\end{proof}

\begin{proof}[Proof of Thm.~\ref{thm:lifting_parallel_to_sequential}]
  We construct the conversion from parallel dilation isometries to sequential dilation unitaries by using similar ideas of Ref.~\cite{tang2025conjugate}.
  We first convert $n$ parallel queries of the isometry channel $\mcV$ into $n$ parallel queries of the random conditional isometry $V(\theta)$ by using the quantum circuit shown in Fig.~\ref{fig:conditional_isometry}~(a) as shown in Lem.~\ref{lemma:parallel_to_sequential_conditional_isometry}.
  Then, we convert parallel $n$ queries of $V(\theta)$ into $n$ copies of the quantum state $V(\theta)V(\theta)^\dagger/d$ by applying them on the maximally mixed state.
  Finally, we convert $n$ copies of the quantum states $V(\theta)V(\theta)^\dagger/d$ to $m$ queries of reflection unitary:
  \begin{align}
    U_\theta &= 2V(\theta)V(\theta)^\dagger-\1_{\mcH},
  \end{align}
  where $\1_\mcH$ is the identity operator on $\mcH$.
  By using the state exponentiation~\cite{lloyd2014quantum,kimmel2017hamiltonian}, we can implement $U_\theta$ by using $O(d^2/\delta)$ copies of $V(\theta)V(\theta)^\dagger/d$ in the approximation error $\delta$ in the diamond norm.
  Thus, we can implement $m$ queries of $U_\theta$ by using $n = m \cdot O(d^2/\delta) = O(m^2 d^2/\varepsilon)$ copies of $V(\theta)V(\theta)^\dagger/d$ with total approximation error $\varepsilon = O(m\delta)$.
  Then, $U_\theta$ implements the desired dilation unitary since
  \begin{align}
    &U_\theta(\ket{0} \otimes V_0 \ket{\psi})\nonumber\\
    &= (e^{i\theta} \ket{0} \otimes V_0 + \ket{1} \otimes V)e^{-i\theta} \ket{\psi} - \ket{0}\otimes V_0\ket{\psi}\\
    &= e^{-i\theta}\ket{1}\otimes V\ket{\psi}
  \end{align}
  holds for all $\ket{\psi}\in \CC^d$
\end{proof}

\section{Main result 2: No-go theorem on sequential random dilation}
\label{sec:sequential_random_dilation}

\begin{figure*}
  \centering
  \includegraphics[width=\linewidth]{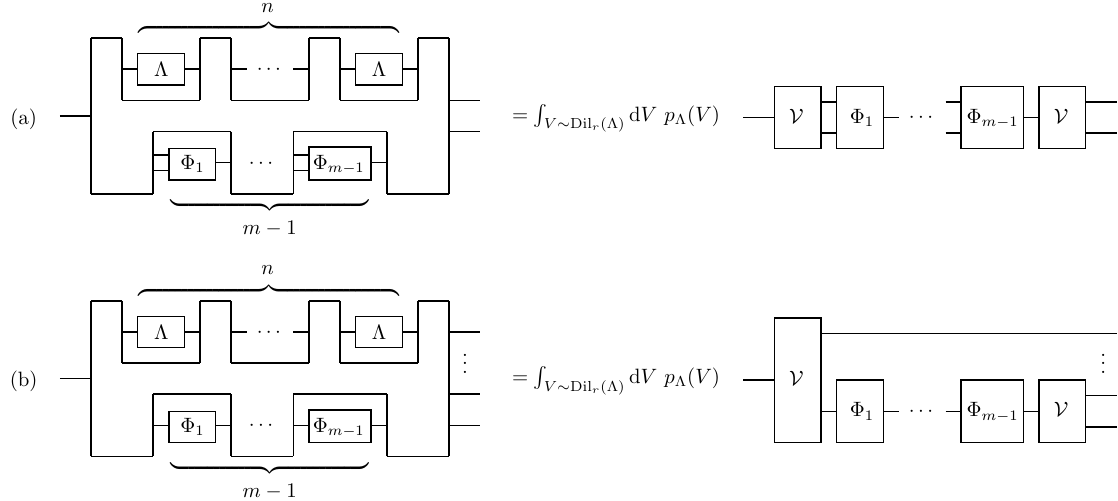}
  \caption{(a) An $n\to m$ random dilation superchannel $\mcS$ is an $(n+m-1)$-slot QC-CC superchannel such that $\mcS(\Lambda^{\otimes n} \otimes \cdot)$ is a probabilistic mixture of $\mcC_V$ with $V\in \mathrm{Dil}_r(\Lambda)$ for any $\Lambda$.
  (b) A weak $n\to m$ random dilation superchannel $\mcS$ is a weak version of an $n\to m$ random dilation superchannel where the environmental systems $\mcE_i$ are not fed back to the inputs of $\Phi_j$ for $i\in [m], j\in [m-1]$.}
  \label{fig:sequential_random_dilation}
\end{figure*}

As a natural extension of the parallel random dilation superchannel, we consider the task of sequential random dilation of quantum channels.
With $n$ queries of an unknown quantum channel $\Lambda: \mcL(\mcI) \to \mcL(\mcO)$, we aim to exactly implement sequential $m$ queries of the dilation isometry $V$ of $\Lambda$ where $V$ is randomly chosen from $\mathrm{Dil}_r(\Lambda)$.
However, we show that such a transformation is impossible for any $m\geq 2$ and $n \leq \max(2, 2^{\lfloor \log_2 \min(d_{\mcI}, d_{\mcO}) \rfloor}-1)$ within the conventional quantum circuit model, where $d_{\mcI} = \dim(\mcI)$ and $d_{\mcO} = \dim(\mcO)$.
To show this, we first define the task of sequential random dilation more rigorously and then show the no-go theorem.

\subsection{Definition of sequential random dilation}
We define the task of sequential random dilation as follows.
The sequential $m$ queries of a dilation isometry $V: \mcI \to \mcE\otimes \mcO$ is represented as a quantum comb $\mcC_V: \bigotimes_{i=1}^{m-1}[\mcL(\mcE_i \otimes \mcO_{i}) \to \mcL(\mcI_{i+1})] \to [\mcL(\mcI_1) \to \mcL(\mcE_m \otimes \mcO_m)]$ with $\mcI_i \simeq \mcI, \mcO_i \simeq \mcO, \mcE_i \simeq \mcE$ for $i\in [m-1]$ defined by
\begin{align}
  \mcC_V(\Phi_1 \otimes \cdots \otimes \Phi_{m-1}) \coloneqq \mcV \circ \Phi_{m-1} \circ \mcV \circ \cdots \circ \Phi_1\circ \mcV
\end{align}
for all linear maps $\Phi_i: \mcL(\mcE_i \otimes \mcO_i) \to \mcL(\mcI_{i+1})$.
We also define a \emph{weakly} sequential $m$ queries of a dilation isometry $V: \mcI \to \mcE\otimes \mcO$ as a quantum comb $\mcC_V^{\mathrm{(weak)}}: \bigotimes_{i=1}^{m-1}[\mcL(\mcO_{i}) \to \mcL(\mcI_{i+1})] \to [\mcL(\mcI_1) \to \mcL(\bigotimes_{j=1}^{m} \mcE_j \otimes \mcO_m)]$ such that
\begin{widetext}
\begin{align}
  \mcC_V^{\mathrm{(weak)}}(\Phi_1 \otimes \cdots \otimes \Phi_{m-1})\coloneqq [\mcV \otimes \1_{\mcE_1\cdots \mcE_{m-1}}] \circ [\Phi_{m-1} \otimes \1_{\mcE_1\cdots \mcE_{m-1}}] \circ [\mcV \otimes \1_{\mcE_1\cdots \mcE_{m-2}}]\circ \cdots \circ [\Phi_1 \otimes \1_{\mcE_1}]\circ \mcV
\end{align}
\end{widetext}
for all linear maps $\Phi_i: \mcL(\mcO_i) \to \mcL(\mcI_{i+1})$.
It is weak in the sense that the outputs in the environment systems $\mcE_i$ are not fed back to the inputs of $\Phi_j$ for $i\in [m], j\in [m-1]$.
An $n\to m$ random dilation superchannel $\mcS$ is an $(n+m-1)$-slot QC-CC superchannel $\mcS$ with the following properties [see also Fig.~\ref{fig:sequential_random_dilation}~(a)]:
\begin{itemize}
  \item The causal order of the last $(m-1)$ slots is fixed.
  \item For any $\Lambda$, $\mcS(\Lambda^{\otimes n} \otimes \cdot)$ is a probabilistic mixture of $\mcC_V$ with $V\in \mathrm{Dil}_r(\Lambda)$, i.e., there exists a probability distribution $\{p_\Lambda(V)\}_{V\in \mathrm{Dil}_r(\Lambda)}$ such that $\mcS(\Lambda^{\otimes n} \otimes \cdot) = \int_{V\in \mathrm{Dil}_r(\Lambda)} \dd V p_\Lambda(V) \mcC_V$.
\end{itemize}
A \emph{weak} sequential $n\to m$ random dilation superchannel $\mcS$ is an $(n+m-1)$-slot QC-CC superchannel $\mcS$ that satisfies the same properties as above but with $\mcC_V$ replaced by $\mcC_V^{\mathrm{(weak)}}$ [see also Fig.~\ref{fig:sequential_random_dilation}~(b)].
By definition, any sequential $n\to m$ random dilation superchannel can be used to construct a weak sequential $n\to m$ random dilation by delaying the outputs in the environmental systems $\mcE_i$.

\subsection{No-go results}

We first show the no-go theorem for $N$-qubit quantum channels.

\begin{theorem}
  \label{thm:multi_qubit_no-go}
  For any $N\geq 1$, there is no weak sequential $n\to 2$ random dilation superchannel for $N$-qubit quantum channels if
  \begin{align}
    n\leq \max(2, 2^N-1).
  \end{align}
\end{theorem}
\begin{proof}
We prove this theorem based on the no-go result on the quantum switch simulation~\cite{bavaresco2025simulating}.
As described in Sec.~\ref{subsec:switch_simulation}, the quantum switch $\mcS_\mathrm{switch}(\Lambda \otimes \Phi)$ cannot be implemented with a QC-CC superchannel with a single query to $\Phi$ and $n\leq \max(2, 2^N-1)$ queries to $\Lambda$.
On the other hand, the quantum channel $\mcS_\mathrm{switch}(\Lambda \otimes\Phi)$ can be implemented by a weak sequential 2 queries of a dilation isometry $V$ of $\Lambda$ and a single query to $\Phi$ by the quantum circuit shown in Fig.~\ref{fig:switch_simulation_by_dilation}.
Thus, if there exists a weak $n\to 2$ random dilation superchannel $\mcS$ for $n \leq \max(2, 2^N-1)$, we can implement $\mcS_\mathrm{switch}(\Lambda \otimes \Phi)$ by a QC-CC superchannel, which contradicts Thm.~\ref{thm:quantum_switch_no-go}.
\end{proof}

\begin{figure}
  \centering
  \includegraphics[width=\linewidth]{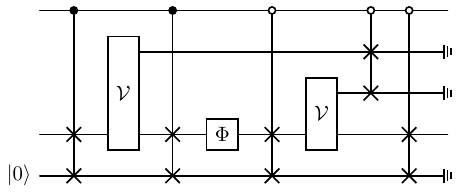}
  \caption{Quantum circuit implementing $\mcS_\mathrm{switch}(\Lambda \otimes \Phi)$ by a weak sequential 2 queries of a dilation isometry $V$ of $\Lambda$ and a single query to $\Phi$.}
  \label{fig:switch_simulation_by_dilation}
\end{figure}

This theorem can be extended to arbitrary-dimensional quantum channels as follows.

\begin{theorem}
  \label{thm:arbitrary_dim_no-go}
  For any $d_{\mcI}, d_{\mcO}\geq 2$, there is no weak $n\to 2$ random dilation superchannel for quantum channels with input dimension $d_{\mcI}$ and output dimension $d_{\mcO}$ if
  \begin{align}
    n\leq \max(2, 2^{\lfloor \log_2 \min(d_{\mcI}, d_{\mcO}) \rfloor}-1).
  \end{align}
\end{theorem}
\begin{proof}[Proof sketch]
We first extend Thm.~\ref{thm:quantum_switch_no-go} to $d$-dimensional quantum channels (i.e., $d_{\mcI}=d_{\mcO}=d$) by extending the no-go result on the quantum switch simulation to $d$-dimensional quantum channels:
\begin{corollary}
  \label{cor:quantum_switch_no-go_arbitrary_dim}
There is no QC-CC superchannel $\mcC$ such that $\mcC(\Lambda^{\otimes n} \otimes \Phi) = \mcS_\mathrm{switch}(\Lambda \otimes \Phi)$ for all $d$-dimensional quantum channels $\Lambda, \Phi$ if
  \begin{align}
    n\leq \max(2, 2^{\lfloor \log_2 d \rfloor}-1).
  \end{align}
\end{corollary}
This can be shown by embedding $\lfloor \log_2 d \rfloor$-qubit quantum channels into $d$-dimensional quantum channels and applying Thm.~\ref{thm:quantum_switch_no-go}.
Then, we extend the no-go result to quantum channels with arbitrary input and output dimensions by embedding $d$-dimensional quantum channels into them for $d=\min(d_{\mcI},d_{\mcO})$.
See Appendix~\ref{appendix:proof_arbitrary_dim_no-go} for the details of the proof.
\end{proof}

\section{Applications}
\label{sec:applications}

As an immediate application of Thm.~\ref{thm:random-dilation-superchannel}, we show that any quantum superchannel with parallel queries of a dilation isometry can be transformed into a quantum superchannel with parallel queries of a quantum channel.

\begin{theorem}
  \label{thm:dilation_isometry_to_channel}
If there exists a quantum superchannel $\mcC$ with $n$ parallel queries of an isometry channel $\mcV$, then there exists a quantum superchannel $\mcC'$ with $n$ parallel queries of a quantum channel $\Lambda:\mcL(\mcI) \to \mcL(\mcO)$ such that
\begin{align}
  \mcC'[\Lambda^{\otimes n}] = \mathbb{E}_{V\sim \mathrm{Dil}_r(\Lambda)}\mcC[\mcV^{\otimes n}].
\end{align}
The circuit complexity of implementing the quantum superchannel $\mcC'$ is given by $n_{\mcC} + O(\mathrm{poly}(n, \log d_{\mcI}, \log d_{\mcO}))$, where $n_{\mcC}$ is the circuit complexity of $\mcC$, $d_{\mcI} = \dim(\mcI)$, and $d_{\mcO} = \dim(\mcO)$.
\end{theorem}

This result generalizes the previous result by Ref.~\cite{chen2025quantum}, which shows a similar statement for the quantum tester corresponding to the case where $\mcI_1 \simeq \CC$ and $\mcO_k$ is the classical system, but without presenting an efficient construction.

We employ Thm.~\ref{thm:dilation_isometry_to_channel} to construct an efficient storage-and-retrieval of an unknown quantum channel $\Lambda: \mcL(\mcI) \to \mcL(\mcO)$ with Kraus rank at most $r$, which is a task to store the action of $\Lambda$ into a quantum state called the program state and to retrieve the action of $\Lambda$ from the program state later only using operations independent of $\Lambda$~\cite{nielsen1997programmable,bisio2010optimal,sedlak2019optimal,sedlak2020probabilistic,yang2020optimal,sedlak2024storage}. 
The size of the program state is called the program cost.
The previous best-known result requires $O(d_{\mcI}^2/\sqrt{\varepsilon})$ queries of $\Lambda$ with the program state in $O(d_{\mcI}^2/\sqrt{\varepsilon})$ qubits, where $\varepsilon$ is the diamond-norm error of the retrieved channel~\cite{yoshida2025quantum}.
If the input channel is an isometry channel, the program cost can be reduced to ${d_{\mcI} d_{\mcO}-1 \over 2} \log O(\varepsilon^{-1})$ qubits with the same query complexity~\cite{yoshida2025quantum}.
By combining this result with Thm.~\ref{thm:random-dilation-superchannel}, we can reduce the program cost for general quantum channels as follows.
This construction shows the exponential reduction of the program cost in $\varepsilon$ compared to the previous best-known result.
\begin{theorem}
  There exists a storage-and-retrieval algorithm of an unknown quantum channel $\Lambda:\mcL(\mcI) \to \mcL(\mcO)$ with Kraus rank at most $r$ by using $O(d_{\mcI}^2 / \sqrt{\varepsilon})$ queries of $\Lambda$ with the program state in ${d_{\mcI} d_{\mcO} r-1\over 2} \log O(\varepsilon^{-1})$ qubits, where $\varepsilon$ is the retrieval error in the diamond norm, $d_{\mcI} = \dim(\mcI)$, and $d_{\mcO} = \dim(\mcO)$.
\end{theorem}
For the case when $d_{\mcI} = d_{\mcO} r$, we can construct the superreplication superchannel and inversion superchannel of an unknown quantum channel $\Lambda$ by combining our random dilation superchannel with the unitary superreplication superchannel and probabilistic unitary inversion shown in Refs.~\cite{chiribella2015universal,quintino2019reversing,quintino2019probabilistic}, respectively (see Fig.~\ref{fig:fig_channel_transformation}).
The former outputs $\Theta(n^{\alpha})$ copies of an unknown quantum channel $\Lambda$ from $n$ copies of $\Lambda$ for any $\alpha<2$ approximately, and the latter outputs the Petz recovery map~\cite{Petz1986sufficientsubalgebas} for the reference state given by the maximally mixed state, probabilistically and exactly.

\begin{figure}
  \centering
  \includegraphics{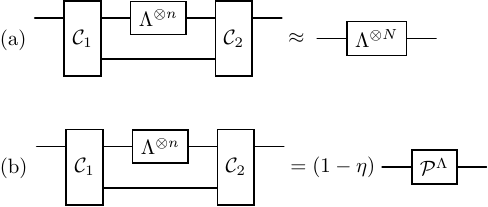}
  \caption{Quantum circuits implementing transformation of an unknown quantum channel $\Lambda:\mcL(\mcI) \to \mcL(\mcO)$ with Kraus rank at most $r = d_{\mcI}/d_{\mcO}$.
  (a) Quantum superchannel that transforms $n$ parallel queries of $\Lambda$ into $N$ approximate parallel queries of $\Lambda$ for $N = \Theta(n^{\alpha})$ with any $\alpha<2$.
  (b) Probabilistic quantum superchannel that transforms $n = O(d_{\mcI}^3/\eta)$ parallel queries of $\Lambda$ into the Petz recovery map $\mcP^\Lambda$ of $\Lambda$ for the reference state given by the maximally mixed state with the success probability at least $1-\eta$.}
  \label{fig:fig_channel_transformation}
\end{figure}

\begin{theorem}
  \label{thm:channel_superreplication}
There exists a quantum superchannel $\mcC$ such that
\begin{align}
  \lim_{n\to\infty} \mathbb{P}\left[\|\mcC[\Lambda^{\otimes n}](\ketbra{\psi}) - \Lambda^{\otimes N} (\ketbra{\psi})\|_1\leq \epsilon\right] = 1
\end{align}
holds for any $N = \Theta(n^{\alpha})$ for $\alpha<2$ and quantum channel $\Lambda:\mcL(\mcI) \to \mcL(\mcO)$ with Kraus rank at most $r = d_{\mcI}/d_{\mcO}$ by using $n$ parallel queries of $\Lambda$, where $\|\cdot\|_1$ is the trace norm, $\epsilon>0$ is an arbitrary constant, and $\mathbb{P}$ represents the probability taken over the Haar measure on the input system $(\CC^{d_{\mcI}})^{\otimes N}$.
The circuit complexity is $O(\mathrm{poly}(n, \log d_{\mcI}))$, where $d_{\mcI} = \dim(\mcI)$ and $d_{\mcO} = \dim(\mcO)$.
\end{theorem}

\begin{proof}
  Reference~\cite{chiribella2015universal} provides a quantum circuit $\mcC'$ transforming $n$ parallel queries of an unknown unitary channel $\mcU(\cdot) = U(\cdot)U^\dagger$ for $U\in \U(d_{\mcI})$ into $N$ approximate parallel queries of $\mcU$ such that\footnote{Although the analysis in Ref.~\cite{chiribella2015universal} is given for the fidelity error, it can be straightforwardly extended to the trace-norm error using the Fuchs--van de Graaf inequality~\cite{fuchs2002cryptographic}.}
  \begin{align}
    \lim_{n\to\infty} \mathbb{P}\left[\|\mcC'[\mcU^{\otimes n}](\ketbra{\psi}) - \mcU^{\otimes N} (\ketbra{\psi})\|_1\leq \epsilon\right] = 1
  \end{align}
  for any $N = \Theta(n^{\alpha})$ for $\alpha<2$.
  By applying Thm.~\ref{thm:dilation_isometry_to_channel} to the quantum circuit $\mcC'$, we obtain the desired quantum circuit $\mcC$ transforming $n$ parallel queries of an unknown quantum channel $\Lambda:\mcL(\mcI) \to \mcL(\mcO)$ with Kraus rank at most $r = d_{\mcI}/d_{\mcO}$ into $N$ approximate parallel queries of $\Lambda$ such that
  \begin{align}
    \mcC''[\Lambda^{\otimes n}] = \mathbb{E}_{U\sim \mathrm{Dil}_r(\Lambda)}\mcC'[\mcU^{\otimes n}].
  \end{align}
  We define a quantum superchannel $\mcC$ by
  \begin{align}
    \mcC = \Tr_{\mcE^{\otimes N}} \circ \mcC''.
  \end{align}
  For any dilation unitary $U: \mcI \to \mcE \otimes \mcO$, since
  \begin{align}
    \Tr_{\mcE^{\otimes N}} [\mcU^{\otimes N}(\ketbra{\psi})] = \Lambda^{\otimes N}(\ketbra{\psi})
  \end{align}
  holds, we have
  \begin{align}
    &\|\mcC[\Lambda^{\otimes n}](\ketbra{\psi}) - \Lambda^{\otimes N} (\ketbra{\psi})\|_1\nonumber\\
    &= \|\Tr_{\mcE^{\otimes N}}\mathbb{E}_{U\sim \mathrm{Dil}_r(\Lambda)}[\mcC'[\mcU^{\otimes n}](\ketbra{\psi}) -\mcU^{\otimes N}(\ketbra{\psi})]\|_1\\
    &\leq \mathbb{E}_{U\sim \mathrm{Dil}_r(\Lambda)}\|\mcC'[\mcU^{\otimes n}](\ketbra{\psi}) -\mcU^{\otimes N}(\ketbra{\psi})\|_1,
  \end{align}
  where we use the data-processing inequality and the convexity of the trace norm in the last inequality.
  Therefore, we obtain the desired result.
  The original unitary superreplication circuit consists of the quantum Schur transforms, and thus we obtain the circuit complexity $O(\mathrm{poly}(n, \log d_{\mcI}))$.
\end{proof}

\begin{theorem}
  \label{thm:channel_petz_recovery}
  There exists a probabilistic quantum superchannel $\{\mcC_S, \mcC_F\}$ such that
  \begin{align}
    \mcC_S[\Lambda^{\otimes n}] &= (1-\eta) \mcP^{\Lambda} (\cdot)
  \end{align}
  for any quantum channel $\Lambda:\mcL(\mcI) \to \mcL(\mcO)$ with Kraus rank at most $r = d_{\mcI}/d_{\mcO}$ by using $n = O(d_{\mcI}^3/\eta)$ parallel queries of $\Lambda$, where $\mcP^{\Lambda}$ is the Petz recovery map of $\Lambda$ with the reference state given by the maximally mixed state.
  The circuit complexity is given by $O(\mathrm{poly}(n))$.
\end{theorem}

\begin{proof}
  References~\cite{quintino2019reversing,quintino2019probabilistic} provide a quantum circuit $\{\mcC_S', \mcC_F'\}$ transforming $n$ parallel queries of an unknown unitary channel $\mcU(\cdot) = U(\cdot)U^\dagger$ for $U\in \U(d_{\mcI})$ into its inverse channel $\mcU^{\dagger}(\cdot) = U^\dagger (\cdot) U$ with success probability at least $1-\eta$ by using $n = O(d_{\mcI}^3/\eta)$ parallel queries of $\mcU$, i.e., $\mcC_S'[\mcU^{\otimes n}] = (1-\eta) \mcU^\dagger$ holds, and $\mcC_F'$ corresponds to the failure branch.
  By applying Thm.~\ref{thm:dilation_isometry_to_channel} to the quantum circuit $\{\mcC_S', \mcC_F'\}$, we obtain the probabilistic quantum superchannel $\{\mcC_S'', \mcC_F''\}$ such that
  \begin{align}
    \mcC_S''[\Lambda^{\otimes n}]
    &= (1-\eta) \mathbb{E}_{U\sim \mathrm{Dil}_r(\Lambda)} \mcU^{\dagger}.
  \end{align}
  For the reference state $\sigma = \pi_{\mcI} \coloneqq {\1_{d_{\mcI}} \over d_{\mcI}}$ given by the maximally mixed state, the Petz recovery map of $\Lambda$ is given by
  \begin{align}
    \mcP^{\Lambda}(\cdot)
    &= \sigma^{1/2} \Lambda^{\dagger}\left(\Lambda(\sigma)^{-{1/2}} (\cdot) \Lambda(\sigma)^{-{1/2}}\right) \sigma^{1/2} = {1 \over r}\Lambda^{\dagger}(\cdot),
  \end{align}
  where we use the relation $\Lambda(\sigma) = \Tr_{\mcE}\mcU(\sigma) = \pi_{\mcO}$ and $\Lambda^\dagger$ is defined by $\Lambda^\dagger(\cdot) = \mcU^\dagger (\1_\mcE \otimes \cdot)$.
  Therefore, by defining the probabilistic quantum superchannel $\{\mcC_S, \mcC_F\}$ as
  \begin{align}
    \mcC_a[\Lambda^{\otimes n}](\cdot) = \mcC_a''[\Lambda^{\otimes n}](\pi_{\mcE} \otimes \cdot) \quad \forall a\in\{S, F\},
  \end{align}
  we obtain
  \begin{align}
    \mcC_S[\Lambda^{\otimes n}](\cdot)
    &= (1-\eta) \mcP^{\Lambda}(\cdot).
  \end{align}
  The original unitary inversion circuit consists of unitary complex conjugation~\cite{miyazaki2019complex} and probabilistic port-based teleportation~\cite{ishizaka2008asymptotic,ishizaka2009quantum,studzinski2017port}.
  The unitary complex conjugation can be implemented by using the quantum Schur transform, and port-based teleportation can be implemented by using the mixed Schur transform as shown in Refs.~\cite{grinko2023gelfand,grinko2023efficient,fei2023efficient,nguyen2023mixed} with the circuit complexity polynomial in $n$ and $d_{\mcI}$.
  Therefore, we obtain the circuit complexity $O(\mathrm{poly}(n, d_{\mcI})) = O(\mathrm{poly}(n))$ since $n = O(\mathrm{poly}(d_{\mcI}))$ holds.
\end{proof}
 
Note that this is the first quantum circuit to exactly implement the Petz recovery map of an unknown quantum channel.
Given access to a dilation unitary, Ref.~\cite{gilyen2022quantum} provides a quantum circuit implementing the Petz recovery map of an unknown quantum channel approximately based on the quantum singular value transformation~\cite{gilyen2019quantum}.
Reference~\cite{utsumi2025quantum} provides a quantum circuit implementing the Petz recovery map of an unknown quantum channel using the Uhlmann transformation approximately, but the approximation error is inevitable in their construction.

\section{Conclusion}
\label{sec:conclusion}
This work provides an efficient quantum circuit transforming $n$ parallel queries of an unknown quantum channel into $n$ parallel queries of a randomly chosen dilation isometry of the channel.
This transformation enables us to generalize various quantum information processing tasks for unitary and isometry channels to those for general quantum channels.
In particular, we show our protocol enables an efficient storage and-retrieval for unknown quantum channels with an exponentially improved program cost in the retrieval error, together with two other applications: superreplication of quantum channels and probabilistic implementation of the Petz recovery map of quantum channels when the Kraus rank of the input quantum channel is at most $d_{\mcI}/d_{\mcO}$.
We also show the extension to the case of quantum superchannels.
Our results open a new avenue for exploring quantum information processing tasks on general quantum channels by leveraging techniques originally developed for unitary or isometry channels.

Although this work positively answers a part of the open problem raised in Ref.~\cite{tang2025conjugate}, our result also provides a negative answer to the other part of the open problem, which is whether there exists a quantum circuit implementing the random dilation superchannel with sequential queries of the input quantum channel.
We show that such an exact sequential random dilation superchannel is impossible with subexponential overhead $o(\mathrm{poly}(\min \{d_\mcI, d_\mcO\}))$ in Thm.~\ref{thm:arbitrary_dim_no-go}.
This no-go result showcases the fundamental difference between parallel and sequential queries of quantum channels.
It also poses an open problem for future work: finding a quantum circuit that implements the \emph{approximate} or \emph{probabilistic} sequential random dilation superchannel.
By transforming the parallel random dilation isometry obtained by Thm.~\ref{thm:random-dilation-superchannel} into sequential random dilation unitaries by Thm.~\ref{thm:lifting_parallel_to_sequential}, we can implement an approximate implementation of the sequential random dilation superchannel.
The port-based teleportation~\cite{ishizaka2008asymptotic,ishizaka2009quantum,studzinski2017port} can transform parallel queries into sequential ones approximately or probabilistically.
However, both of them incur $O(\mathrm{poly}(d_{\mcI}))$ overhead in the number of queries to achieve constant success probability or approximation error.
It remains open whether we can construct a sequential random dilation superchannel with $O(\mathrm{poly}(\log d))$ overhead in the number of queries to achieve a constant success probability or approximation error.

\acknowledgments
We acknowledge helpful discussions with Jessica Bavaresco, Marco Fanizza, Filippo Girardi, Zixuan Liu, Antonio Anna Mele, Francesco Anna Mele, Marco T\'{u}lio Quintino, and Elias Theil.
This work was supported by Japan Society for the Promotion of Science (JSPS) KAKENHI Grant Numbers 23KJ0734, 23K21643, JP24K16975, JP25K00924, and JP26H02015, the MEXT Quantum Leap Flagship Program (MEXT QLEAP) JPMXS0118069605, JPMXS0120351339, JST CREST Grant Number JPMJCR25I5, JPMJCR23I3, JST ASPIRE Grant Number JPMJAP25A3, JST NEXUS Grant Number JPMJNX26C9, JPMJNX26C2, JST SPRING Grant Number JPMJSP2108, FoPM, WINGS Program, the University of Tokyo, DAIKIN Fellowship Program, the University of Tokyo, and IBM Quantum.

\section*{Note added}
During the preparation of this work, we became aware of an independent work~\cite{girardi2025random2}, which also constructs a quantum circuit implementing the random dilation superchannel shown in Fig.~\ref{fig:random-dilation-superchannel}~(a) of this work.

\bibliography{main}
\clearpage

\appendix
\onecolumngrid

\section{Another construction of random dilation superchannel}
\label{appendix:another_construction_random_dilation_superchannel}
Here we present the proof that the quantum circuit shown in Fig.~\ref{fig:channel_random_dilation_kronecker} implements the random dilation superchannel.
We apply the proof idea of Thm.~\ref{thm:random-dilation-superchannel}, which is to convert the random purification channel to the random dilation superchannel by using the P-CTC.
By applying this idea to the random purification channel shown in Ref.~\cite{tang2025conjugate}, we obtain the quantum circuit shown in Fig.~\ref{fig:fig_channel_random_dilation_kronecker}~(a).
By decomposing the quantum Schur transform $U_\mathrm{Sch}^{(n, d_{\mcI} d_{\mcO})}$ into $U_\mathrm{Sch}^{(n, d_{\mcI})}$ and $U_\mathrm{Sch}^{(n, d_{\mcO})}$ using the Kronecker transform, we obtain the quantum circuit shown in Fig.~\ref{fig:fig_channel_random_dilation_kronecker}~(b).
By bending the wires in the P-CTCs, we obtain the quantum circuit shown in Fig.~\ref{fig:fig_channel_random_dilation_kronecker}~(c), where the control on $\ket{i_\mu}$ means that the Young diagram $\mu$ is extracted from $\ket{i_\mu}$, and it is used to control the gate $\mathrm{CG}_{\mu\nu}$.
Note that the standard encoding of $i_\mu$ has $\mu$ in the last system~\cite{burchardt2025high}.
In Fig.~\ref{fig:fig_channel_random_dilation_kronecker}~(c), a P-CTC only appears after the inverse Kronecker transform $\mathrm{CG}_{\mu \nu}^\dagger$, and this P-CTC can be removed by considering the following relation on the Clebsch--Gordan coefficients [see Eq. (10), p. 298 in Ref.~\cite{klimyk1995representations}]:
\begin{align}
  \bra{i_\mu} \bra{a} \mathrm{CG}_{\lambda\nu} \ket{i_{\lambda}} \ket{i_{\nu}} = \sqrt{m_\mu \over m_\lambda} \bra{i_\lambda} \bra{a} \mathrm{CG}_{\mu\nu} \ket{i_\mu} \ket{i_\nu}.
\end{align}
From this expression, the inverse Kronecker transform $\mathrm{CG}_{\mu \nu}^\dagger$ can be replaced with the inverse Kronecker transform $\mathrm{CG}_{\lambda\nu}^\dagger$ without P-CTCs, and we obtain the quantum circuit shown in Fig.~\ref{fig:channel_random_dilation_kronecker}.

\begin{figure}
    \centering
    \includegraphics[width=\linewidth]{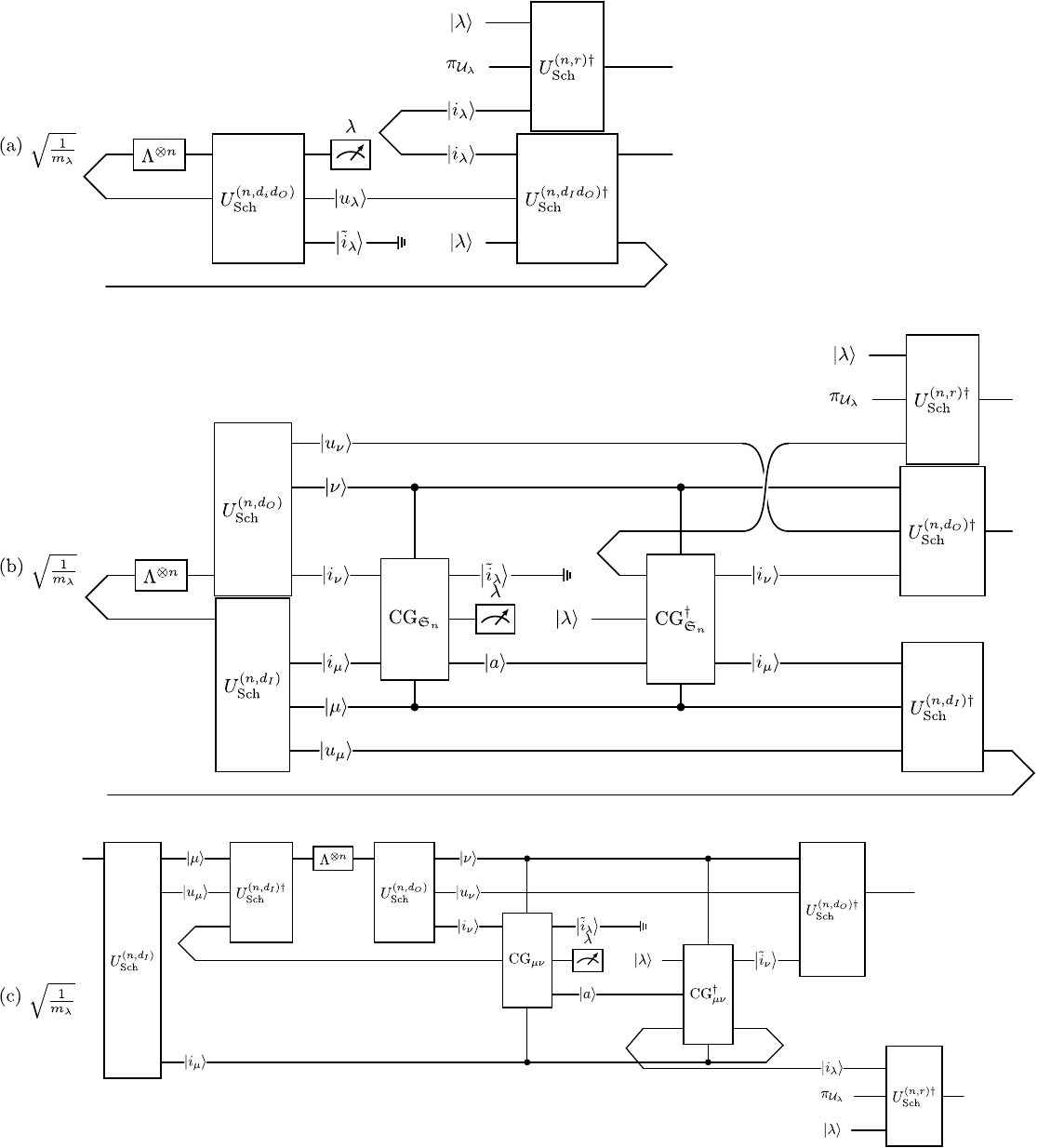}
    \caption{The proof idea to get another construction of random dilation superchannel shown in Fig.~\ref{fig:channel_random_dilation_kronecker}.}
    \label{fig:fig_channel_random_dilation_kronecker}
\end{figure}

\section{Proof of Theorem~\ref{thm:arbitrary_dim_no-go}}
\label{appendix:proof_arbitrary_dim_no-go}
We first provide the proof of Cor.~\ref{cor:quantum_switch_no-go_arbitrary_dim}.
\begin{proof}[Proof of Cor.~\ref{cor:quantum_switch_no-go_arbitrary_dim}]
  For any $d, d'$ satisfying $d<d'$, we define the embedding isometry $E_{d\to d'}:\CC^d \to \CC^{d'}$ and the restricting channel $\mcR_{d'\to d}:\mcL(\CC^{d'}) \to \mcL(\CC^{d})$ as follows:
  \begin{align}
    E_{d\to d'} &\coloneqq \sum_{i=0}^{d-1} \ketbra{i}{i},\\
    \mcR_{d'\to d}(\cdot) &\coloneqq E_{d\to d'}^\dagger (\cdot) E_{d\to d'} + \Tr[(\1_{d'} - E_{d\to d'} E_{d\to d'}^\dagger)(\cdot)] \ketbra{0}{0},
  \end{align}
  using the computational basis $\{\ket{i}\}$.
  Then, we can embed $\lfloor \log_2 d \rfloor$-qubit quantum channels $\Lambda', \Phi'$ into a $d$-dimensional quantum channels $\Lambda, \Phi$ as follows:
  \begin{align}
    \Lambda = \mcE_{2^{\lfloor \log_2 d \rfloor}\to d}  \circ \Lambda' \circ \mcR_{d\to 2^{\lfloor \log_2 d \rfloor}},\quad \Phi =  \mcE_{2^{\lfloor \log_2 d \rfloor}\to d}\circ \Phi' \circ  \mcR_{d\to 2^{\lfloor \log_2 d \rfloor}},
  \end{align}
  where $\mcE_{d\to d'}$ is the embedding channel defined by $\mcE_{d\to d'}(\cdot) = E_{d\to d'} (\cdot) E_{d\to d'}^\dagger$.
  Then, we have $(\1_2 \otimes \mcR_{d\to 2^{\lfloor \log_2 d \rfloor}})\circ \mcS_\mathrm{switch}(\Lambda \otimes \Phi)\circ (\1_2 \otimes \mcE_{2^{\lfloor \log_2 d \rfloor}\to d}) = \mcS_\mathrm{switch}(\Lambda', \Phi')$, which implies that if there exists a QC-CC superchannel $\mcS$ such that $\mcS(\Lambda^{\otimes n} \otimes \Phi) = \mcS_\mathrm{switch}(\Lambda \otimes \Phi)$ for all $d$-dimensional quantum channels $\Lambda, \Phi$ with $n\leq \max(2, 2^{\lfloor \log_2 d \rfloor}-1)$, we can implement $\mcS_\mathrm{switch}(\Lambda', \Phi')$ by a QC-CC superchannel for all $\lfloor \log_2 d \rfloor$-qubit quantum channels $\Lambda', \Phi'$, which contradicts Thm.~\ref{thm:quantum_switch_no-go}.
\end{proof}

\begin{proof}[Proof of Thm.~\ref{thm:arbitrary_dim_no-go}]
Corollary~\ref{cor:quantum_switch_no-go_arbitrary_dim} shows Thm.~\ref{thm:arbitrary_dim_no-go} for the case of $d_{\mcI}=d_{\mcO}$.
We show Thm.~\ref{thm:arbitrary_dim_no-go} by contradiction, i.e., we assume that there exists a sequential random dilation superchannel $\mcC$ for quantum channels with input dimension $d_{\mcI}$ and output dimension $d_{\mcO}$ with $n\leq \max(2, 2^{\lfloor \log_2 \min(d_{\mcI}, d_{\mcO}) \rfloor}-1)$, and we show that this assumption leads to a contradiction.
We can embed $d$-dimensional quantum channels $\Lambda'$ into quantum channels $\Lambda: \mcL(\CC^{d_{\mcI}}) \to \mcL(\CC^{d_{\mcO}})$ as follows:
\begin{align}
  \Lambda = 
  \begin{cases}
    \mcE_{d_{\mcI} \to d_{\mcO}} \circ \Lambda' & (d_{\mcI}\leq d_{\mcO})\\
    \Lambda'\circ \mcR_{d_{\mcI} \to d_{\mcO}} & (d_{\mcI}>d_{\mcO})
  \end{cases}.
\end{align}
From this observation, we define a QC-CC superchannel $\mcC'$ by
\begin{align}
  \mcC'(\Lambda'^{\otimes n} \otimes \Phi_1\otimes \cdots \otimes \Phi_m) \coloneqq
  \begin{cases}
    (\1_\mcE \otimes \mcR_{d_{\mcO} \to d_{\mcI}}) \circ \mcC(\mcE_{d_{\mcI}\to d_{\mcO}}^{\otimes n} \circ \Lambda'^{\otimes n} \otimes \Phi_1\circ \mcR_{d_{\mcO} \to d_{\mcI}}\otimes \cdots \otimes \Phi_{m-1}\circ \mcR_{d_{\mcO} \to d_{\mcI}}) & (d_{\mcI}\leq d_{\mcO})\\
    \mcC(\Lambda'^{\otimes n} \circ \mcR_{d_{\mcI} \to d_{\mcO}}^{\otimes n} \otimes \mcE_{d_{\mcO} \to d_{\mcI}} \circ \Phi_1 \otimes \cdots \otimes \mcE_{d_{\mcO} \to d_{\mcI}} \circ \Phi_{m-1}) \circ \mcE_{d_{\mcO \to d_{\mcI}}} & (d_{\mcI}>d_{\mcO})
  \end{cases}.
\end{align}
Then, $\mcC'(\Lambda'\otimes \cdot)$ implements the sequential random dilation of $\Lambda'$ for all $d$-dimensional quantum channels $\Lambda'$.
Thus, we conclude that there exists a QC-CC superchannel that implements the sequential random dilation of $\Lambda'$ for all $d$-dimensional quantum channels $\Lambda'$, which contradicts Cor.~\ref{cor:quantum_switch_no-go_arbitrary_dim}.
\end{proof}

\end{document}